	\documentclass[
		paper=a4,11pt,american,pagesize,%
		abstract=true,DIV=12,headinclude=true,headlines=0,footinclude=true,footlines=-5,twoside=semi,%
	]{scrartcl}
	\def\docclass{koma}
	\def\version{arxiv}
	\def\draftmode{false} %
\usepackage[utf8]{inputenc}
\usepackage{hyperref}
\makeatletter
\usepackage{xifthen}

\newcommand\iflipics[2]{\ifthenelse{\equal{\docclass}{lipics}}{#1}{#2}}
\newcommand\ifkoma[2]{\ifthenelse{\equal{\docclass}{koma}}{#1}{#2}}
\newcommand\ifieee[2]{\ifthenelse{\equal{\docclass}{ieee}}{#1}{#2}}
\newcommand\ifsiam[2]{\ifthenelse{\equal{\docclass}{siam}}{#1}{#2}}
\newcommand\ifsiamsingle[2]{\ifthenelse{\equal{\docclass}{siam-single}}{#1}{#2}}
\newcommand\ifmysiam[2]{\ifthenelse{\equal{\docclass}{my-siam}}{#1}{#2}}
\newcommand\ifacm[2]{\ifthenelse{\equal{\docclass}{acm}}{#1}{#2}}
\newcommand\ifdcc[2]{\ifthenelse{\equal{\docclass}{dcc}}{#1}{#2}}
\newcommand\ifspringerjournal[2]{\ifthenelse{\equal{\docclass}{springer-journal}}{#1}{#2}}
\newcommand\iflncs[2]{\ifthenelse{\equal{\docclass}{lncs}}{#1}{#2}}
\ifthenelse{ \equal{\docclass}{lipics} \OR \equal{\docclass}{koma} \OR \equal{\docclass}{ieee} \OR \equal{\docclass}{siam} \OR \equal{\docclass}{my-siam} \OR \equal{\docclass}{acm} \OR \equal{\docclass}{dcc} \OR \equal{\docclass}{springer-journal} \OR \equal{\docclass}{lncs} }{
	\PackageInfo{paper}{Building paper with docclass = \docclass} % all good
}{
	\PackageWarning{paper}{docclass = "\docclass", but must be one of "lipics", "koma", "ieee", "siam", "siam-single", "my-siam", "acm", "dcc", "springer-journal", "lncs"}
}

\newcommand\ifmanuscript[2]{\ifthenelse{\equal{\version}{manuscript}}{#1}{#2}}
\newcommand\ifarxiv[2]{\ifthenelse{\equal{\version}{arxiv}}{#1}{#2}}
\newcommand\ifsubmission[2]{\ifthenelse{\equal{\version}{submission}}{#1}{#2}}
\newcommand\ifproceedings[2]{\ifthenelse{\equal{\version}{proceedings}}{#1}{#2}}
\ifthenelse{ 
	\equal{\version}{manuscript} 
	\OR \equal{\version}{arxiv} 
	\OR \equal{\version}{submission} 
	\OR \equal{\version}{proceedings} 
}{
	\PackageInfo{paper}{Building paper version = \version} % all good
}{
	\PackageWarning{paper}{version = "\version", but must be one of "manuscript", "arxiv", "submission", "proceedings"}
}

\newcommand\ifdraft[2]{\ifthenelse{\equal{\draftmode}{true}}{#1}{#2}}
\ifthenelse{ \equal{\draftmode}{true} \OR \equal{\draftmode}{false} }{
	\PackageInfo{paper}{Building paper with draftmode = \draftmode} % all good
}{
	\PackageWarning{paper}{draftmode = "\draftmode", but must be "true" or "false"}
}

\usepackage[T1]{fontenc}
\ifsiam{
	\usepackage{lmodern}
	\usepackage{slantsc}
}{}
\ifsiamsingle{
	\usepackage{lmodern}
	\usepackage{slantsc}
}{}
\ifmysiam{
	\usepackage{lmodern}
	\usepackage{slantsc}
}{}
\ifkoma{
	\usepackage{lmodern}
	\usepackage{slantsc}
}{}
\iflipics{
	\usepackage{lmodern}
	\usepackage{slantsc}
}{}
\ifdcc{
	\usepackage{lmodern}
	\usepackage{slantsc}
}{}
\ifspringerjournal{
	\usepackage{lmodern}
	\usepackage{slantsc}
	\usepackage{xcolor}
}{}
\iflncs{
	\usepackage{lmodern}
	\usepackage{slantsc}
	\usepackage{xcolor}
}{}

\usepackage{babel}
\input{ushyphex.tex} % American English hyphenation exceptions

\usepackage{array,multicol}
\ifieee{
	\usepackage[cmex10]{amsmath,mathtools}
	\usepackage{amsfonts,amssymb}
}{}
\ifkoma{
	\usepackage{amsmath,amsfonts,amssymb,mathtools}
}{}
\iflipics{
	\usepackage{amsmath,amsfonts,amssymb,mathtools}
}{}
\ifsiam{
	\usepackage{amsmath,amsfonts,amssymb,mathtools}
}{}
\ifsiamsingle{
	\usepackage{amsmath,amsfonts,amssymb,mathtools}
}{}
\ifmysiam{
	\usepackage{amsmath,amsfonts,amssymb,mathtools}
}{}
\ifacm{
	\usepackage{mathtools}
}{}
\ifdcc{
	\usepackage{amsmath,amsfonts,amssymb,mathtools}
}{}
\ifspringerjournal{
	\usepackage{amsmath,amsfonts,amssymb,mathtools}
}{}
\iflncs{
	\usepackage{amsmath,amsfonts,amssymb,mathtools,comment}
}{}

\usepackage{mleftright}\mleftright % redefine left and right to not add space around braces
\usepackage{relsize,xspace,booktabs,adjustbox,needspace,pbox,relsize}
\ifieee{
		
		\usepackage{enumitem}
}{
	\ifspringerjournal{
		\usepackage{enumitem}
		\setlist{topsep=\medskipamount}
	}{
		\usepackage{enumitem}
	}
}
\usepackage{graphicx}

\ifacm{}{
	\usepackage{colonequals}
}

\iflncs{
	\usepackage[american]{wref}
}{
	\usepackage{wref}
}

\newdimen\makeboxdimen
\ifthenelse{\equal{\docclass}{koma} \OR \equal{\docclass}{my-siam}}{
	\setlength\parindent{1.5em}
	\usepackage[headsepline]{scrlayer-scrpage}
	\pagestyle{scrheadings}
	\clearscrheadfoot
	\AtBeginDocument{%
		\automark[section]{}%
	}
	\ohead{\pagemark}
	\rehead{\mytitle}
	\lohead{\headmark}
	\addtokomafont{caption}{\sffamily\small}
	\addtokomafont{captionlabel}{\sffamily\textbf}
	\setcapmargin{2em}
}{}
\ifmysiam{
	\setcapmargin{1em}
	\setcapindent{0em}
}{}
\ifmysiam{
	\setlength\parskip{0pt}
	\RedeclareSectionCommand[
		beforeskip=-1.25\baselineskip,
		afterskip=0.75\baselineskip,
	]{section}
	\RedeclareSectionCommand[
		beforeskip=-1\baselineskip,
		afterskip=-1.5em,
	]{subsection}
	\RedeclareSectionCommand[
		beforeskip=-1\baselineskip,
		afterskip=-1.5em,
	]{subsubsection}
	\RedeclareSectionCommand[
		beforeskip=-.25\baselineskip,
		indent=1.5em,
		afterskip=-1em,
	]{paragraph}
}{}
\ifdcc{
	\ifproceedings{}{
		\pagestyle{plain}
		\setlength{\footskip}{5ex}
	}
}{}
\iflncs{
	\ifproceedings{}{
		\pagestyle{plain}
	}
}{}

\AtBeginDocument{%
	\let\mytitle\@title%
}

\ifsiam{
	\usepackage[bibtex-alpha]{url-doi-arxiv}
}{}
\ifsiamsingle{
	\usepackage[bibtex-alpha]{url-doi-arxiv}
}{}
\ifmysiam{
	\usepackage[bibtex-alpha]{url-doi-arxiv}
}{}
\ifkoma{
	\usepackage[bibtex-alpha]{url-doi-arxiv}
}{}
\iflipics{
	\usepackage[bibtex]{url-doi-arxiv}
}{}
\ifdcc{
	\usepackage[bibtex-alpha]{url-doi-arxiv}
}{}
\ifspringerjournal{
	\usepackage[bibtex-alpha]{url-doi-arxiv}
}{}
\iflncs{
	\bibliographystyle{splncs04} % We encourage you to use ... so not enforced
}{}

\let\oldthebibliography\thebibliography
\renewcommand\thebibliography[1]{%
	\oldthebibliography{#1}%
	\pdfbookmark[1]{References}{}%
}

\usepackage{lscape} % sideways figures

\ifkoma{
	\usepackage{float}
	\floatstyle{boxed}
	\usepackage{newfloat}
}{}
\iflipics{
	\usepackage{newfloat}
	\DeclareFloatingEnvironment[%
			name=Algorithm,%
			placement=thb,%
		]{algorithm}
}{}
\ifmysiam{
	\usepackage{newfloat}
	\DeclareFloatingEnvironment[%
			name=Algorithm,%
			placement=thb,%
		]{algorithm}
}

\ifdcc{
	\usepackage{float}
	\usepackage[font={small},labelfont=bf]{caption}
}{}
\ifmysiam{
	\setcounter{topnumber}{3}
	\setcounter{bottomnumber}{3}
	\setcounter{totalnumber}{3}     % 2 may work better
	\setcounter{dbltopnumber}{3}    % for 2-column pages
}{}
\ifsiam{
	\setcounter{topnumber}{3}
	\setcounter{bottomnumber}{3}
	\setcounter{totalnumber}{3}     % 2 may work better
	\setcounter{dbltopnumber}{3}    % for 2-column pages
}{}
\ifsiamsingle{
	\setcounter{topnumber}{3}
	\setcounter{bottomnumber}{3}
	\setcounter{totalnumber}{3}     % 2 may work better
}{}

\usepackage{dcolumn}

\usepackage{textcomp} % needed for upquote
\usepackage{listings}

\usepackage{tikz}

\usetikzlibrary{positioning,arrows.meta,fit}
\usetikzlibrary{backgrounds,calc,trees,graphs}
\usetikzlibrary{shapes.geometric,shapes.misc}

\pgfdeclarelayer{background}
\pgfsetlayers{background,main}

\usetikzlibrary{external}
\tikzsetexternalprefix{pics/externalized/}
\tikzset{
	external/system call={%
		lualatex \tikzexternalcheckshellescape -halt-on-error %
			-interaction=batchmode -jobname "\image" "\texsource"%
	},
}
\tikzset{external/export=false} % Only export manually marked pictures, enable for costly pictures
\iflipics{
	\newtheorem{fact}[theorem]{Fact}

	\newenvironment{proofof}[1]{%
		\begin{proof}[{{Proof of #1{}}}]%
	}{%
		\end{proof}%
	}
}{}
\ifacm{
	\AtEndPreamble{
		\theoremstyle{acmdefinition}
		\newtheorem{remark}[theorem]{Remark}
		\newtheorem{fact}[theorem]{Fact}
	}
	
	\newenvironment{proofof}[1]{%
		\begin{proof}[{{Proof of #1{}}}]%
	}{%
		\end{proof}%
	}
}{}
\ifsiam{
	\newtheorem{remark}{Remark}
	\newenvironment{proofof}[1]{%
		\begin{proof}[{{#1{}}}]%
	}{%
		\end{proof}%
	}
}{}
\ifsiamsingle{
	\newtheorem{remark}{Remark}
	\newenvironment{proofof}[1]{%
		\begin{proof}[{{#1{}}}]%
	}{%
		\end{proof}%
	}
}{}
\iflncs{
	\spnewtheorem{fact}[theorem]{Fact}{\itshape}{}
	
	\let\orig@endproof\endproof
	\def\endproof{\qed\orig@endproof}
	\newenvironment{proofof}[1]{%
		\begin{proof}[{{#1{}}}]%
	}{%
		\end{proof}%
	}
}{}
\ifthenelse{%
		\equal{\docclass}{lipics} \OR \equal{\docclass}{siam} \OR 
		\equal{\docclass}{siam-single} \OR \equal{\docclass}{acm} \OR
		\equal{\docclass}{lncs}%
}{}{
	\usepackage[amsmath,hyperref,thmmarks]{ntheorem}
	
	\theorembodyfont{\slshape}
	\theoremseparator{:}
	\newtheoremstyle{proofstyle}%
	  {\item[\theorem@headerfont\hskip\labelsep ##1\theorem@separator]}%
	  {\item[\theorem@headerfont\hskip\labelsep ##3\theorem@separator]}
	
	\theorempreskip{\topsep} % new version of ntheorem

	\theoremsymbol{\adjustbox{scale=.8}{$\triangleleft\mkern-1mu$}}
	
	\newtheorem{theorem}{Theorem}[section]
	
	\theoremstyle{plain}
	\theorempreskip{\topsep}

	\newtheorem{lemma}[theorem]{Lemma}

	\theoremstyle{plain}
	\theorembodyfont{\upshape}

	\newtheorem{remark}[theorem]{Remark}
	\newtheorem{example}[theorem]{Example}
	
	\theoremsymbol{\raisebox{-.25ex}{$\Box$}}
	\qedsymbol{\raisebox{-.25ex}{$\Box$}}
	
	\theoremstyle{proofstyle}
	\newtheorem{proof}{Proof}
	\newenvironment{proofof}[1]{%
		\begin{proof}[{{Proof of #1{}}}]%
	}{%
		\end{proof}%
	}

}

\iflipics{
		\newenvironment{thmenumerate}[2][]{%
			\begin{enumerate}[
				label={\textsf{\textbf{\color{darkgray}{\makebox[\widthof{(a)}][c]{\textup{(\alph*)}}}}}},
				ref={\ref{#2}\kern.1em--\kern.1em(\alph*)},
				itemsep=0pt,
				topsep=.5ex,
				leftmargin=1.75em,
				#1
			]%
		}{%
			\end{enumerate}%
		}
}{
	\ifspringerjournal{
		\newenvironment{thmenumerate}[2][]{%
			\begin{enumerate}[
				label={\makebox[\widthof{(a)}][c]{\textup{(\alph*)}}},
				ref={\ref{#2}\kern.1em--\kern.1em(\alph*)},
				itemsep=0pt,
				topsep=\smallskipamount,
				leftmargin=1.75em,
				#1
			]%
		}{%
			\end{enumerate}%
		}
	}{
		\newenvironment{thmenumerate}[2][]{%
			\begin{enumerate}[
				label={\makebox[\widthof{(a)}][c]{\textup{(\alph*)}}},
				ref={\ref{#2}\kern.1em--\kern.1em(\alph*)},
				itemsep=0pt,
				#1
			]%
		}{%
			\end{enumerate}%
		}
	}
}

\newcommand*\ie{\mbox{i.\hspace{.2ex}e.}}
\newcommand*\eg{\mbox{e.\hspace{.2ex}g.}}

\newcommand\N{\mathbb N}

\usepackage{fixmath}

\newcommand{\ESymbol}{\mathbb{E}}

\newcommand{\ProbSymbol}{\ensuremath{\mathbb{P}}}
\providecommand{\given}{}
\DeclarePairedDelimiterXPP\Prob[1]{\ProbSymbol}[]{}{%
	\renewcommand\given{\nonscript\:\delimsize\vert\nonscript\:\mathopen{}}%
	#1%
}
\DeclarePairedDelimiterXPP\E[1]{\ESymbol}[]{}{%
	\renewcommand\given{\nonscript\:\delimsize\vert\nonscript\:\mathopen{}}%
	#1%
}
\DeclarePairedDelimiterXPP\Eover[2]{\ESymbol_{#1}}[]{}{%
	\renewcommand\given{\nonscript\:\delimsize\vert\nonscript\:\mathopen{}}%
	#2%
}
\DeclarePairedDelimiterXPP\ProbIn[2]{\ProbSymbol_{#1}}[]{}{%
	\renewcommand\given{\nonscript\:\delimsize\vert\nonscript\:\mathopen{}}%
	#2%
}
\providecommand{\Prob}{} % hack for syntax highlighting ...
\providecommand{\ProbIn}{} % hack for syntax highlighting ...
\providecommand{\E}{} % hack for syntax highlighting ...
\providecommand{\Eover}{} % hack for syntax highlighting ...

\newcommand{\surroundedmath}[3]{% #1=mathrel/mathbin/etc #2=spacing #3=symbol
	\mathchoice{%display
		#1{#2{#3}#2}%
	}{%text
		#1{#3}%
	}{%script
		#1{#3}%
	}{%scriptsript
		#1{#3}%
	}%
}
\newcommand\rel[1]{\surroundedmath{\mathrel}{\:}{#1}}
\newcommand\wrel[1]{\surroundedmath{\mathrel}{\;}{#1}}
\newcommand\wwrel[1]{\surroundedmath{\mathrel}{\;\;}{#1}}
\newcommand\bin[1]{\surroundedmath{\mathbin}{\:}{#1}}
\newcommand\wbin[1]{\surroundedmath{\mathbin}{\;}{#1}}

\newcommand\ppe{\phantom{=}}

\iflipics{}{
\iflncs{}{
	\makeatletter
	\let\oldalign\align
	\let\endoldalign\endalign
	\renewenvironment{align}{%
		\begingroup%
		\let\oldhalign\halign
		\def\halign{%
			\let\oldbreak\\%
			\def\nonnumberbreak{\nonumber\oldbreak*}%
			\def\\{%
				\@ifstar{\nonnumberbreak}{\oldbreak}%
			}%
			\oldhalign%
		}
		\oldalign%
	}{%
		\endoldalign%
		\endgroup%
	}
}}
\newcommand*\numberthis[1][]{\stepcounter{equation}\tag{\theequation}}

\allowdisplaybreaks[3]

\newcommand\splitaftercomma[1]{%
  \begingroup
  \begingroup\lccode`~=`, \lowercase{\endgroup
    \edef~{\mathchar\the\mathcode`, \penalty0 \noexpand\hspace{0pt plus .25em}}%
  }\mathcode`,="8000 #1%
  \endgroup
}

\def\mydots{\xleaders\hbox to.5em{\hfill.\hfill}\hfill}
\newlength\tmpLenNotations

\ifdraft{%
	\iflipics{
		\usepackage{lineno} 
	}{
		\usepackage[switch]{lineno} 
	}
	\linenumbers
	\overfullrule=6mm
	
	\usepackage[color,notref,notcite]{showkeys}
	\definecolor{refkey}{gray}{.99}
	\colorlet{labelkey}{green!60!black!60}
	
	\usepackage[inline,nolabel]{showlabels}

	\showlabels{cite}
	\showlabels{citealt}
	\showlabels{citealp}
	\showlabels{citet}
	\showlabels{Citet}
	\showlabels{citep}
	\showlabels{citeauthor}
	\showlabels{Citeauthor}
	\showlabels{citefullauthor}
	\showlabels{citeyear}
	\showlabels{citeyearpar}
	\showlabels{wref}
	\showlabels{wpref}
	\showlabels{wtpref}
	\showlabels{wildref}
	\showlabels{wildpageref}
	\showlabels{wildtpageref}
}{}

\iflipics{
	\ifmanuscript{\hideLIPIcs}{}
	\ifarxiv{\hideLIPIcs}{}
	\ifsubmission{}{\nolinenumbers}
}{}

\ifdraft{}{%
	\usepackage{microtype}
}

\hypersetup{
	final,
	unicode=true, 
	bookmarks=true,
	bookmarksnumbered=true,
	bookmarksdepth=2,
	bookmarksopen=true,
	breaklinks=true,
	hidelinks,
}

\newsavebox\tmpbox

\ifdcc{
	\renewcommand\paragraph{\@startsection{paragraph}{4}{\parindent}%\z@}%
	                                      {\smallskipamount}%{3.25ex \@plus1ex \@minus.2ex}%
	                                      {-1em}%
	                                      {\normalfont\normalsize\bfseries}}
}{}
\iflipics{
	\let\oldparagraph\paragraph
	\renewcommand\paragraph[1]{%
		\oldparagraph*{#1}
	}
}{
	\let\oldparagraph\paragraph
	\renewcommand\paragraph[1]{%
		\oldparagraph{\boldmath #1.}
	}
}

\ifmysiam{
	\let\oldsubsection\subsection
	\renewcommand\subsection[1]{%
		\oldsubsection{#1.}%
	}
	\let\oldsubsubsection\subsubsection
	\renewcommand\subsubsection[1]{%
		\oldsubsubsection{#1.}%
	}
}{}
\ifsiam{
	\let\oldsubsection\subsection
	\renewcommand\subsection[1]{%
		\oldsubsection{#1.}%
	}
	\let\oldsubsubsection\subsubsection
	\renewcommand\subsubsection[1]{%
		\oldsubsubsection{#1.}%
	}
}{}
\ifsiamsingle{
	\let\oldsubsection\subsection
	\renewcommand\subsection[1]{%
		\oldsubsection{#1.}%
	}
	\let\oldsubsubsection\subsubsection
	\renewcommand\subsubsection[1]{%
		\oldsubsubsection{#1.}%
	}
}{}

\let\epsilon\varepsilon

\def\myacknowledgements{}
\ifkoma{
	
}{}
\ifieee{
	
}{}
\ifsiam{
	
}{}
\ifsiamsingle{
	
}{}
\ifmysiam{
	
}{}
\ifdcc{
	
}{}
\ifacm{
	
}{}
\ifspringerjournal{
	
}{}
\iflncs{
	
}{}

\newcommand*\rankop{\mathsf{rank}}
\newcommand*\selop{\mathsf{select}}
\newcommand*\accessop{\mathsf{access}}

\newcommand*\DegreeEntropy{\ensuremath{H_0^{\deg}}\xspace}
\newcommand*\StringOfTargets{\ensuremath{A}\xspace}

\newcommand*\DegreeInGraph{\ensuremath{d}\xspace}
\newcommand*\InDegreeInGraph{\ensuremath{d_{\text{in}}}\xspace}
\newcommand*\OutDegreeInGraph{\ensuremath{d_{\text{out}}}\xspace}
\newcommand*\OutNeighborhood{\ensuremath{N_{\text{out}}}\xspace}

\newcommand\DAGof[1]{\ensuremath{\mathit{DAG}(#1)}\xspace}
\newcommand\UDAGof[1]{\ensuremath{S(\mathit{DAG}(#1))}\xspace}

\newcommand*\MultinomialDist{\ensuremath{\mathrm{Mult}}\xspace}
\newcommand\PrefAttch[1][M;n]{\ensuremath{\mathrm{PA}\left(#1\right)}\xspace}

\makeatother

\usepackage[Operation]{algorithm}
\usepackage{algorithmicx}

\usepackage[indLines=true]{algpseudocodex}
\usepackage{alphabeta} % for appendix

\newcommand\arxivpaper{extended arXiv version (\url{TODO})\xspace}

	\title{Succinct Preferential-Attachment Graphs}
	\newcommand\email[1]{\texttt{#1}}
	\author{%
			Ziad Ismaili Alaoui%
			\footnote{University of Liverpool, UK,
				\email{\{ziad.ismaili-alaoui, n.namrata, sebastian.wild\}\,@\,liverpool.ac.uk}}
		\and
			Namrata\footnotemark[1]%
	\and
		Sebastian Wild%
			\footnotemark[1]\,
			\footnote{University of Marburg, Germany, 
			\email{wild\,@\,informatik.uni-marburg.de}}
	}
	
	\date{\small\today}

\begin{document}

\maketitle
%} %

%

\begin{abstract}
Computing over compressed data combines the space saving of data compression 
with efficient support for queries directly on the compressed representation.
Such data structures are widely applied in text indexing and have been successfully generalised to trees.
For graphs, support for computing over compressed data remains patchy;
typical results in the area of succinct data structures 
are restricted to a specific class of graphs and use the same, 
worst-case amount of space for any graph from this class.

In this work, we design a data structure whose space usage
automatically improves with the compressibility of the graph at hand,
while efficiently supporting navigational operations (simulating adjacency-list access).
Specifically, we show that the space usage approaches the \emph{instance-optimal} space
when the graph is drawn according to the classic Barabási-Albert model of preferential-attachment graphs.
Our data-structure techniques also work for arbitrary graphs,
guaranteeing a size asymptotically no larger than an entropy-compressed edge list.
A~key technical contribution is the careful analysis of the instance-optimal space usage.
\ifproceedings{
	\par(Omitted proofs are found online in the \arxivpaper of this paper.)
}{}
\end{abstract}

\section{Introduction}

In this paper, we design a compressed representation for graphs generated according to the Barabási-Albert model of preferential attachment approaching the instance-optimal space usage of $\lg(1/p)$ bits, for $p$ the probability of the stored graph, that supports efficient operations within the same space.
(Here and throughout, $\lg = \log_2$.)

The motivation and techniques of our work span the fields of information theory, compression algorithms, and succinct data structures. 
Information theory studies random processes (\emph{sources}) for generating combinatorial objects, as a way to quantify the intrinsic information content in an object $x$. Compression methods are algorithms to efficiently construct representations of $x$ that come close to this lower bound for the size of representations.
Succinct data structures aim to support efficient queries for an object $x\in \mathcal X$ using $\lg|\mathcal X|(1+o(1))$ bits of space. This space usage is asymptotically optimal in the \emph{worst case} over $\mathcal X$, but can, in principle, be improved to $\lg(1/\Prob x)$ bits of space when the object is drawn randomly from $\mathcal X$ with probability $\Prob x$. \emph{Universal (source) codes}, as studied in information theory, approximate that space usage over a whole class of random sources ``automatically'', \ie, \emph{without knowing} the actual probabilities $\Prob x$.
In some restricted cases (namely strings~\cite{FerraginaManzini2000} and trees~\cite{MunroNicholsonSeelbachBenknerWild2021}), such universal codes have successfully been augmented with efficient query support. 
When applied to graphs, information theory, compression algorithms, and succinct data structures all are, by comparison, still in their infancy. 
A short survey of existing work on graphs from these angles is given in \wref{sec:related-work}.

A classical notion of compressibility of a text $T$ is its (zeroth-order) \emph{empirical entropy} $H_0(t)$, (formally defined in \wref{sec:preliminaries}). $H_0(T)$ gives a lower bound for the size of any representation that independently encodes individual characters of $T$, such as the well-known Huffman codes.
It is also the entropy rate of the maximum-likelihood memoryless source for $T$, or the entropy of the character distribution obtained by Bayesian inference, \ie, starting with a uniform prior distribution, 
for each $T[i]$,
update our prior (to a Beta distribution) to make $T[i]$ more likely to appear in the future.

\begin{figure}[tb]
	\centering
	\def\probsize{\smaller[1]}
	\begin{tikzpicture}[scale=.8]
		\tikzset{
			unlabeled vertex/.style={circle, draw, fill, minimum size=5pt,inner sep=0pt},
			directed edge/.style={thick,draw,{-Stealth[]}},
			undirected edge/.style={thick,draw},
			emph/.style={red},
		}
		\useasboundingbox (-.3,-1) rectangle (14.6,1.5) ;
		\begin{scope}[shift={(0,0)}]
			\foreach \v/\s in {0/,1/{emph}} {
				\node[unlabeled vertex,\s] (v\v) at (\v,0) {} ;
			}
			\foreach \f/\t/\s in {%
					1/0/{},1/0/{bend left},1/0/{bend right}%
			} {
				\ifthenelse{\equal{\f}{1}}{
					\draw[directed edge,emph] (v\f) to [\s={20*(\t-\f)-10}] (v\t) ;
				}{
					\draw[directed edge] (v\f) to [\s={20*(\t-\f)-10}] (v\t) ;
				}
			}
			\node at (0.5,1.3) {\probsize $t=1$};
			\node at (.5,-.9) {\probsize $1$};
		\end{scope}
		\begin{scope}[shift={(2.5,0)}]
			\foreach \v/\s in {0/,1/,2/{emph}} {
				\node[unlabeled vertex,\s] (v\v) at (\v,0) {} ;
			}
			\foreach \f/\t/\s in {%
					1/0/{},1/0/{bend left},1/0/{bend right},%
					2/1/{},2/0/{bend left},2/0/{bend right}%
			} {
				\ifthenelse{\equal{\f}{2}}{
					\draw[directed edge,emph] (v\f) to [\s={20*(\t-\f)-10}] (v\t) ;
				}{
					\draw[undirected edge] (v\f) to [\s={20*(\t-\f)-10}] (v\t) ;
				}
			}
			\node at (1,1.3) {\probsize $t=2$};
			\node at (1,-.9) {\probsize $3\cdot \frac{3}{6}\cdot\frac{3}{6}\cdot \frac{3}{6}$};
		\end{scope}
		\begin{scope}[shift={(6,0)}]
			\foreach \v/\s in {0/,1/,2/,3/{emph}} {
				\node[unlabeled vertex,\s] (v\v) at (\v,0) {} ;
			}
			\foreach \f/\t/\s in {%
					1/0/{},1/0/{bend left},1/0/{bend right},%
					2/1/{},2/0/{bend left},2/0/{bend right},%
					3/1/{bend left},3/0/{bend left},3/1/{bend right}%
			} {
				\ifthenelse{\equal{\f}{3}}{
					\draw[directed edge,emph] (v\f) to [\s={20*(\t-\f)-10}] (v\t) ;
				}{
					\draw[undirected edge] (v\f) to [\s={20*(\t-\f)-10}] (v\t) ;
				}
			}
			\node at (1.5,1.3) {\probsize $t=3$};
			\node at (1.5,-.9) {\probsize $3\cdot \frac{5}{12}\cdot\frac{4}{12}\cdot \frac{4}{12}$};
		\end{scope}
		\begin{scope}[shift={(10.5,0)}]
			\foreach \v/\s in {0/,1/,2/,3/,4/{emph}} {
				\node[unlabeled vertex,\s] (v\v) at (\v,0) {} ;
			}
			\foreach \f/\t/\s in {%
					1/0/{},1/0/{bend left},1/0/{bend right},%
					2/1/{},2/0/{bend left},2/0/{bend right},%
					3/1/{bend left},3/0/{bend left},3/1/{bend right},%
					4/1/{bend left},4/0/{bend right},4/3/{}%
			} {
				\ifthenelse{\equal{\f}{4}}{
					\draw[directed edge,emph] (v\f) to [\s={20*(\t-\f)-10}] (v\t) ;
				}{
					\draw[undirected edge] (v\f) to [\s={20*(\t-\f)-10}] (v\t) ;
				}
			}
			\node at (2,1.3) {\probsize $t=4$};
			\node at (1.95,-.9) {\probsize $6\cdot \frac{6}{18}\cdot\frac{6}{18}\cdot \frac{3}{18}$};
		\end{scope}
	\end{tikzpicture}
	\caption{%
		The Barabási-Albert \PrefAttch model of iteratively growing graphs with $M=3$
		for $n=4$ time steps. The probability of the shown choices of $M$ targets are given for each time step, 
		yielding $\Prob{G_4} = 5/864$ overall. The information content is thus $\lg(1/\Prob{G_4}) \approx 7.43$ bits; the degree-entropy is $\DegreeEntropy(G_4)\approx 15.37$.
	}
	\label{fig:barabasi-albert}
\end{figure}
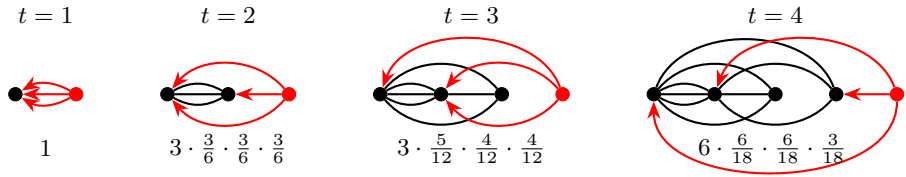

Several competing notions of empirical entropies for trees have been proposed~\cite{HuckeLohreySeelbachBenkner2020,HuckeLohreySeelbachBenkner2019,KiefferYangSzpankowski2009} (see also \cite[Part\,II]{SeelbachBenkner2023}); for graphs, the grounding in information theory is weaker still.
A simple and intuitive notion of empirical entropy for (directed) graphs $G$ can be obtained as follows:
Concatenate the adjacency lists (the out-neighbours) of all vertices and consider the empirical entropy $H_0(\StringOfTargets)$ of the resulting \emph{adjacency string} $\StringOfTargets = \StringOfTargets(G)$ over the alphabet $V(G)$. 
Observe that each vertex $v\in V(G)$ appears exactly $\InDegreeInGraph(v)$ times (for $\InDegreeInGraph(v)$ the in-degree of $v$) in $\StringOfTargets(G)$; hence we call this quantity $\DegreeEntropy(G)$, the \emph{degree entropy} of $G$.
For ease of presentation, we focus on graphs where the out-degree of all vertices is the same and a canonical order of the vertices is given; then it is easy to reconstruct $G$ from $\StringOfTargets(G)$.
Note that $\DegreeEntropy(G) = H_0(\StringOfTargets(G)) \le m \lg n$, where the latter is the number of bits used by an uncompressed adjacency list.
(Here $n = |V(G)|$ is the number of vertices and $m=|E(G)|$ the number of edges.)

We will show that despite its simplicity, $\DegreeEntropy(G)$ can be a meaningful benchmark for graph compression.
We point out here that this is at least partially surprising since $\StringOfTargets(G)$ fixes a specific ordering (and hence, naming) of vertices and edges.
We (almost always) consider the graph unchanged when (1) the edges are listed in different order, and, in many applications, (2) when arbitrary renaming of vertices occurs.
In the latter case, we only seek to preserve the \emph{structure} of a graph (\ie, the \emph{unlabelled graph} or the equivalence class under relabelling), \eg, to encode labels separately.
In addition, \emph{undirected} edges are given a direction by listing them in $\StringOfTargets(G)$.
So, in general $\StringOfTargets^{-1}(G)$ is far from unique and $\DegreeEntropy(G)$ is \emph{not} in general a lower bound for the description length for $G$.
This complication is indicative of what makes graph compression harder than text compression.

As a challenging testbed for compressed graph data structures, we hence consider \emph{unlabelled}, \emph{undirected} graphs from a highly non-uniform distribution. 
Specifically, we assume the \emph{Barabási-Albert model}~\cite{BarabasiAlbert1999} of random \emph{preferential-attachment graphs}.
It generates a graph iteratively as follows (\wref{fig:barabasi-albert}). At each time step $t$, we add one new vertex $v_t$
and draw $M$ random neighbours for $v_t$, where $M$ is a fixed parameter of the model.
The $M$ neighbours $a_{t,1},\ldots,a_{t,M}$ are chosen, independently and with repetitions allowed, from the existing vertices $v_0,\ldots,v_{t-1}$ with probability \emph{proportional} to the (current) \emph{degree} of vertex $v_i$, $\Prob{u_{t,j} = v_i} \propto d_t(v_i)$. 
This rule follows the \emph{Matthew principle} (rich vertices get richer) and produces graphs where the degree-distribution follows a power law (asymptotically, the probability of a vertex to collect total degree $k$ is proportional to $k^{-3}$).
The heavy-tailed power-law degree distribution is widely observed in application domains with (approximately) scale-free graphs and makes compression interesting and non-trivial.
The Barabási-Albert model is also a natural analogue of Bayesian inference of character distributions:
whenever we see a vertex as the target of an edge, we consider this target as more likely for future choices.

\subsection{Our Results}

Suppose $G$ is a graph generated by the Barabási-Albert model, and let $\Prob{G}$ be the probability
for $G$ to arise in this process (where we assume that $M$ is a fixed parameter and we stop the generation process after the $n$th vertex).
Any encoding for graphs that is simultaneously optimal for all such graphs must use (close to) $\lg(1/\Prob{G})$ bits of space to encode $G$.
Our first set of results relate the amount of information in $G$, $\lg(1/\Prob{G})$, to the empirical degree-entropy of $G$, $\DegreeEntropy(G)$: despite their seemingly unrelated origins, the two quantities \emph{coincide} up to an error term.

\begin{theorem}[Instance-specific lower bound]
\label{thm:lower-bound-lbl}
	Let $G$ be a labelled graph generated by preferential attachment, $G\sim\PrefAttch$,
	and further assume that apart from a fixed-size seed graph, the graph~$G$ is simple.
	Then
	\[
		\lg(1/\Prob{G}) \wwrel= \DegreeEntropy(G) \bin\pm O(n M \lg M).
	\]
\end{theorem}

Note that $O(n M \lg M)$ is a lower order term for typical graphs 
when $M$ is not too big, that is $\lg M = o(\lg n)$.
The analysis (\wref{sec:lower-bound}) uses Gibbs' inequality to compare different notions of empirical entropy, a trick that might be of independent interest.

For \emph{unlabelled} graphs $G$, the bound reduces by $n\lg n$ bits:

\begin{theorem}[Unlabelled lower bound]
\label{thm:lower-bound-ulbl}
	Let $G$ be as in \wref{thm:lower-bound-lbl} and consider its structure $S = S(G)$,
	\ie, a random \emph{unlabelled} graph generated by preferential attachment.
	Then
	\[
		\lg(1/\Prob{S}) \wwrel\ge \DegreeEntropy(S) - n \lg n \bin- O(n M \lg M).
	\]
\end{theorem}

(Note that $\DegreeEntropy(G) = \DegreeEntropy(S(G))$.)
It may seem obvious that the bound would reduce precisely by the $\lg(n!)\sim n\lg n$ bits needed to store $n$ vertex labels
since any unlabelled graph can correspond to at most $n!$ labelled graphs;
however, we point out that 
(1) for general distributions over labelled graphs, it is not true that these $n!$ labelled graphs are equally likely, and
(2) a strictly higher lower bound is the actual truth for certain $S$, \eg, graphs with linear diameter.
We obtain a tight bound
\ifsubmission{%
    in \wref{app:omitted-proofs},
}{%
    in \wref{sec:lower-bound},
}%
but its general growth with $n$ is opaque.
We do not know whether the bound in \wref{thm:lower-bound-ulbl} is tight.
\ifproceedings{%
	Due to space constraints, some proofs are deferred to the extended arXiv version of this paper.
}{}

The degree entropy thus precisely captures the asymptotic information content of a pref\-e\-ren\-tial-attachment graph.
We next show that we can represent $G$ using close to $\DegreeEntropy(G)$ bits of space 
and support efficient queries on the compressed representation.

\begin{theorem}[Ultrasuccinct preferential-attachment graphs]
\label{thm:data-structure}
	Let $G$ be obtained from $\PrefAttch$.
	There are data structures that support operations
        finding 
        (1) the $i$th in-neighbour, 
        (2) the $i$th out-neighbour, 
        (3) the degree of a vertex, and 
        (4) testing adjacency,
        all in $O(\lg n)$ time,
    in the following total space:
	\begin{thmenumerate}{thm:data-structure}
	\item \label{thm:data-structure-labeled} 
		Storing the labelled graph $G$ uses $\DegreeEntropy(G) + o(Mn)$ bits of space.
	\item \label{thm:data-structure-unlabeled} 
		Storing the unlabelled graph $S=S(G)$ uses at most $\DegreeEntropy(S)(1-\frac1M) + 2n + o(Mn) \le (M-1)n\lg n + 2n + o(Mn)$ bits of space.
	\end{thmenumerate}
\end{theorem}
By \wref{thm:lower-bound-lbl}, the space for (a) is asymptotically optimal.
The space for (b) almost matches \wref{thm:lower-bound-ulbl}; in particular, it does so on average:
$\smash{\frac1M} \DegreeEntropy(S)$ is close to $n \lg n$ unless $S$ is very compressible.
(Note that $A(S)$ has $Mn$ characters, so $\DegreeEntropy(S) = Mn \cdot H_0^{\text{pc}}(A(S))$
and $\frac1M \DegreeEntropy(S) = n\cdot H_0^{\text{pc}}(A(S))$, for $H_0^{\text{pc}}$ the \emph{per-character entropy}; $H_0^{\text{pc}}(A(S))$ is always between $0$ and $\lg n$.)

The labelled graph data structure uses an entropy-compressed wavelet tree to represent the string $A(G)$.
To reduce the space for an unlabelled graph, our data structure uses one outgoing edge per vertex (except $v_0$) as a \emph{parent edge} in an ordinal tree. These can be stored using a succinct tree data structure using only 2 bits per edge 
if we use the tree indices to identify the graph vertices. 
The remaining out-neighbours form string $A'$, which replaces $A$ and is now $n$ entries shorter, saving up to $\lg n$ bits each.
By judiciously choosing which edges to turn into tree edges, we can moreover \emph{retain} the \emph{compressibility} of $A'$;
we show that our choice implies that the per-character entropy in $A'$ is never larger than that of $A$%
\ifproceedings{ (Appendix~C of the extended arXiv version).}{.}

The unlabelled-graph data structure assigns new labels to the vertices of the graph, which have to be used in queries; 
the relabelling function can be provided during construction, but cannot be stored as part of the data structure in the stated space.
A use case for the unlabelled-graph data structure is any task where we only need to identify a small subset of vertices; \eg, computing network-analysis metrics such as betweenness centrality for a subset of important vertices.
We then store the names of important vertices in addition to the data structure for $S(G)$.

\subsection{Related Work}
\label{sec:related-work}

The closest work to ours is the analysis of \PrefAttch by \L uczak et al.~\cite{LuczakMagnerSzpankowski2019,LuczakMagnerSzpankowski2019isit}.
They compute the entropy of the distribution $G\sim\PrefAttch$, $\E{\lg(1/\Prob{G})} = (M-1) n\lg n \cdot (1+o(1))$  and (only in the conference version~\cite{LuczakMagnerSzpankowski2019isit}) describe a compression algorithm achieving this entropy in expectation up to a redundancy of $O(n \lg\lg n)$.
Our work improves their compression method to redundancy $O(n)$ (using a succinct tree instead of their backtracking numbers), provides instance-specific bounds, and supports efficient operations on the compressed representation.  (It is not clear whether their compression algorithm can be turned into an efficient compressed data structure.)

Many more specialized models of random graphs have been suggested and studied.
An example that also features a compressed data structure is the geometric inhomogeneous random graph (GIRG) model~\cite{BringmannKeuschLengler2019}.
The expected compressed space matches the entropy up to constant factors.  
It is not clear if instance-optimal space is easily achievable. 
As stated, the compression method requires a geometric realization of the graph as input.

Another related area is applied graph compression, where methods are evaluated empirically on benchmarks from specific domains.  As a representative example, we mention the webgraph framework~\cite{BoldiVigna2004}, which provides space-efficient data structures tuned to large web graphs.

The idea to use wavelet trees for the edges of a graph is discussed in Navarro's survey on wavelet trees~\cite{navarro2014wavelet} (and may be folklore). This approach alone is inherently labelled and thus cannot escape the lower bound of \wref{thm:lower-bound-lbl}.
The closest prior work to our data structure is the GLOUDS graph representation of Fischer and Peters~\cite{FischerPeters2016}.
GLOUDS partitions the edges into tree edges and non-tree edges using a standard breadth-first search to obtain an ordinal tree; it uses a custom representation for the tree, as well.
GLOUDS does not achieve optimal space for $G\sim \PrefAttch$ with constant $M$. 
In its original form it also does not adapt to $\DegreeEntropy(G)$ (although entropy compressing their array $H$ would move towards the latter).

There are many further works that are related in spirit to our work, but with a distinctly different angle.
For certain specific classes of graphs, \emph{succinct} data structures have been designed:
interval graphs \cite{AcanChakrabortyJoRao2020,HeMunroNekrichWildWu2020},
chordal graphs~\cite{MunroWu2018},
permutation graphs~\cite{TsakalidisWildZamaraev2023},
circle graphs~\cite{AcanChakrabortyJoNakashimaSadakaneRao2022},
bounded clique-width graphs~\cite{ChakrabortyJoSadakaneRao2024},
series-parallel graphs~\cite{ChakrabortyJoSadakaneRao2023}.
For representing any graph from a class of graphs with $g_n$ graphs on vertex set $[n]$, these data structures use $\lg(g_n)(1+o(1))$ bits of space for representing one such graph;
this invariably corresponds to a uniform distribution over the class of graphs under consideration.
Supported operations vary, but always include finding neighbours of a given vertex.

For trees, data structures beating the worst-case optimal space on compressible inputs have been considered.
Ultrasuccinct trees~\cite{JanssonSadakaneSung2012} adapt to the degree entropy of a tree. 
Hypersuccinct trees~\cite{MunroNicholsonSeelbachBenknerWild2021} have been shown to simultaneously yield optimal space for a large variety of entropy measures, both empirical (such as degree entropy) and entropy rates of random tree sources.
Here we lift these results to ``ultrasuccinct graphs''.

On the information theory of structures (unlabelled graphs), a line of work studied the entropy of random unlabelled graphs.
Apart from the works of \L uczak et al.~\cite{LuczakMagnerSzpankowski2019,LuczakMagnerSzpankowski2019isit}
on preferential attachment graphs, small-world graphs~\cite{KontoyiannisLimPapakonstantinopoulouSzpankowski2022}, the Erdős-Renyi graphs~\cite{ChoiSzpankowski2012}, and a vertex-copying model~\cite{TurowskiMagnerSzpankowski2020} have been studied.

Lossless compression of graphs is surveyed by Besta and Hoefler~\cite{BestaHoefler2019};
see also~\cite{BouritsasLoukasKaraliasBronstein2021}.
Typical examples consider an application and exploit specifics of graphs arising there. 
A statistical model of data and efficient queries are often not available.

\section{Preliminaries}
\label{sec:preliminaries}

In this section, we collect some definitions and known results that we build on in this work.
We write $[a..b]$ for $\{a,a+1,\ldots, b\}$ and abbreviate $[n] = [1..n]$; also $[a..b) = [a..b-1]$.
We use standard graph terminology; specifically, for a graph $G$ and a vertex $v\in V(G)$, we write $\DegreeInGraph(v) = \DegreeInGraph_G(v)$ for the degree of $v$ in $G$, \ie, the number of (potentially parallel) edges incident at $v$.
For directed graphs (digraphs), we distinguish in-degree and out-degree: $\DegreeInGraph(v) = \InDegreeInGraph(v) + \OutDegreeInGraph(v)$; we also write $\OutNeighborhood(v)$ for the out-neighbourhood of $v$, \ie, the (multi)set of vertices $w$ for which there is a (bundle of) edge(s) $vw$ in $E(G)$. 
We call a DAG \emph{$M$-out-regular} if $\OutDegreeInGraph(v) = M$ for all vertices $v\in V\setminus\{v_0\}$.

\subsection{Empirical Entropy}

For a text $T[1..n]$ over alphabet $\Sigma = [1..\sigma]$, the \emph{zeroth-order empirical entropy} $H_0(T)$ is given by
\ifsubmission{%
\(
		H_0(T)
	\wwrel=
		\sum_{c=1}^{\sigma} |T|_c \lg\bigl( \frac{n}{|T|_c} \bigr),
\)
}{%
\[
		H_0(T)
	\wwrel=
		\sum_{c=1}^{\sigma} |T|_c \lg\left( \frac{n}{|T|_c} \right),
\]
}%
where $|T|_c = |\{i \in [1..n] : T[i] =c\}|$ denotes the number of occurrences of $c$ in~$T$. Additionally, we define $H_0^{\mathrm{pc}}(T)=H_0(T)/|T|$ to be the \emph{per-character zeroth-order empirical entropy} of $T$. As is standard, we set $0 \lg (n/0) := 0$ for notational convenience.

\subsection{Preferential-Attachment Graphs}

We consider the classical Barabási-Albert model of generating a random (undirected) graph (with parallel edges); see \wref{fig:barabasi-albert}.
Given a target size $n$ and a parameter $M\in\N$, we grow a graph $G$ iteratively,
where vertex $v_t$ arrives at time $t$, for $t=0,\ldots,n$.
We denote by $G_t$ the graph after this arrival at time $t$; the vertex set is $V(G_t) = \{v_0,\ldots,v_{t}\}$.
$G_0$ is the single isolated vertex $v_0$; $G_1$ has $M$ parallel edges between $v_0$ and $v_1$.
$G_t$ for $t\ge 2$ results from $G_{t-1}$ by adding $v_t$ and $M$ edges from $v_t$ to targets $a_{t,1},\ldots,a_{t,M} \in V(G_{t-1}) = \{v_0,\ldots,v_{t-1}\}$, where the
$a_{t,1},\ldots,a_{t,M}$ are mutually independent and identically distributed with 
\[
		\Prob{a_{t,1} = v_\ell} 
	\wrel=\cdots\wrel= 
		\Prob{a_{t,M} = v_\ell} 
	\wwrel= 
		\frac{\DegreeInGraph_{G_{t-1}}(v_\ell)}{2|E(G_{t-1})|}
		\qquad(\ell=0,\ldots,t-1).
\]
Note that if we orient the edges from new to old vertices in $G_n\sim\PrefAttch$ (in the order they arrived to the graph, \ie, from large to small $t$), we obtain a directed acyclic graph with uniform out-degrees: ${\OutDegreeInGraph}(v_t) = M$ for $t\in[1..n]$ ($\OutDegreeInGraph(v_0)=0$). 

We will abbreviate $\DegreeInGraph_t(v) = \DegreeInGraph_{G_{t}}(v)$.
Note that $|E(G_{t})| = tM$ and $|V(G_t)| = t+1$.
We write $G_n\sim\PrefAttch$ to indicate that $G_n$ has been randomly chosen according to this preferential-attachment process.

For $G \sim \PrefAttch$, we obtain $\DegreeEntropy(G)$ as 
$\DegreeEntropy(G) = H_0(A(G))$ where 
\[A(G) = [a_{1,1},\ldots,a_{1,M},\,\ldots,\, a_{n,1},\ldots,a_{n,M}].\]

Let $S(G)$ denote the \emph{structure} of $G$, \ie, the class of graphs which are isomorphic to~$G$. 
We also refer to these as \emph{unlabelled graphs}.
We define the unlabelled graph distribution $\mathit{PA}^u(M;n)$ as the probability distribution on the family of~$S(G)$, where the probability of each class is the sum of the probabilities that a labelled version of $S(G)$ is $\PrefAttch$, that is,
\ifsubmission{\(}{\[}
\sum_{H\in S(G)} \Prob{H = \PrefAttch}.
\ifsubmission{\)}{\]}

\subsection{Succinct Data Structures}

For the reader's convenience, we collect used results
on succinct data structures here.
First, we cite the compressed bit vectors of Pătrașcu~\cite{Patrascu2008}.

\begin{lemma}[Compressed bit vector]
\label{lem:compressed-bit-vectors}
	Let $B[1..n]$ be a bit vector of length~$n$, containing $m$ $1$-bits.
	For any constant $c>0$, there is a data structure using
	\(
			\lg \binom{n}{m} \wbin+ O\bigl(\frac{n}{\lg^c n}\bigr)
		\wwrel\le 
			m \lg \bigl(\frac nm\bigr) \wbin+ O\bigl(\frac{n}{\lg^c n}+m\bigr)
	\)
	bits of space that
	supports in $O(1)$ time operations 
	(for $i \in [1..n]$):
	\begin{itemize}
		\item $\accessop(B, i)$: return $B[i]$, the bit at index $i$ in $B$;
		\item $\rankop_\alpha(B, i)$: return the number of bits with
		value $\alpha \in \{0,1\}$ in $B[1..i]$;
		\item $\selop_\alpha(B, i)$: return the index of the $i$th
		bit with value $\alpha \in \{0,1\}$.
	\end{itemize}
\end{lemma}

Using wavelet trees, we can support rank and select queries on
arbitrary static strings/sequences while compressing them to zeroth-order empirical entropy.

\begin{lemma}[Wavelet tree~\cite{navarro2014wavelet}]
\label{lem:wavelet-tree-simple}
	Let $S[1..n]$ be a array with entries $S[i]\in\Sigma = [1..\sigma]$. There is a data structure using $H_0(S) + o(n)$ bits of space that supports the following queries in $O(\lg \sigma)$ time (without access to $S$ at query time):
    \begin{itemize}

	\item $\accessop(S, i)$: return $S[i]$, the symbol at index $i$ in $S$;
	\item $\rankop_\alpha(S, i)$: return the number of indices with value $\alpha \in \Sigma$ in $S[1..i]$;
	\item $\selop_\alpha(S, i)$: return the index of the $i$th occurrence of value $\alpha \in \Sigma$ in $S$.
    \end{itemize}
\end{lemma}

\ifsubmission{}{
	\begin{proof}
	Many variants of wavelet trees are possible~\cite{navarro2014wavelet}; 
	for our result, we use the plain encoding of the characters with $\lceil \lg\sigma \rceil$ bits each
	(balanced wavelet tree).  All levels except the last are full, so we can concatenate all bitvectors
	of nodes in the wavelet tree in level order (pointerless wavelet tree~\cite[\S2.3]{navarro2014wavelet})
	and support navigation up and down the tree.
	Moreover, using \wref{lem:compressed-bit-vectors} for this bitvector of length at most $\lceil \lg\sigma \rceil n$ yields space $H_0(S) + o(n)$ (\cite[\S3.1]{navarro2014wavelet}).
	\end{proof}
}

We lastly need a succinct data structure for ordinal trees.
Several options exist for the basic queries we need~\cite{BenoitDemaineMunroRamanRamamRao2005,GearyRamanRaman2006}; (see \cite[App.\,A]{HeMunroNekrichWildWu2020} for a comprehensive review).

\begin{lemma}[{Succinct ordinal trees}]
    \label{lem:tree-succinct}
    Let $T$ be an ordinal tree on $n$ vertices.
    There is a data structure using $2n+o(n)$ bits of space that supports the following queries in $O(1)$ time 
    (where nodes are identified with their preorder index in $T$):
    \begin{itemize}
        \item $\mathsf{parent}(T,v)$: return the parent of $v$ in $T$;
        \item $\mathsf{degree}(T,v)$: return the number of children of $v$ in $T$;
        \item $\mathsf{child}(T,v,i)$: return the $i$th child of $v$ in $T$.
    \end{itemize}
\end{lemma}

\section{Space Lower Bound}
\label{sec:lower-bound}

In this section, we derive the instance-specific lower bounds for preferential-attachment graphs.
\wref{sec:proof-lower-bound-lbl} gives the main result for \emph{labelled} graphs.
\ifproceedings{%
	The appendix of the \arxivpaper
}{%
	\wref{app:proof-lower-bound-ulbl} 
}%
gives the proof for the \emph{unlabelled} case.

Let $G \sim \PrefAttch$ be a labelled preferential-attachment graph. Recall that~$G$ is drawn iteratively, where vertex~$v_t$ arrives at time~$t$, for~$t = 1, \ldots, n$. The graph after the arrival of~$v_t$ is~$G_t$, where $V(G_t) = \{v_0, \ldots, v_t\}$. The total number of edges in~$G = G_{n}$ is~$m = Mn$.

For time~$t = 1, \ldots, n$, let~$\OutNeighborhood(v_t)$ denote the random multiset of size~$M$ that contains the~$M$ out-neighbours of~$v_t$ randomly sampled from the set~$V(G_{t-1})  = \{v_0, v_1, \ldots, v_{t-1}\}$. The vertex~$v_t$ attaches to a vertex in~$V(G_{t-1})$  with probability proportional to its degree at time~$t-1$:
\begin{equation}
\label{eq:prob_dist}
		P^{(t)}_{i} 
	\wrel= 
		\Prob[\big]{v_i \in \OutNeighborhood(v_t) \given G_{t-1}} 
	\wwrel= 
		\frac {d_{t-1}(v_i)} {2(t-1)M} \;.
\end{equation}

\paragraph{Multinomial distribution of~$\OutNeighborhood(v_t)$} 
The frequencies of appearance of vertices as out-neighbours of~$v_t$ 
follow a multinomial distribution. 
The number of trials in this case is~$M$, where each trial when~$v_t$ chooses the neighbour from the set~$V(G_{t-1})$ is independent of each other. 
Also, each trial has~$t$ possible outcomes $v_0, v_1, \ldots v_{t-1}$, with probabilities as given by Equation~\eqref{eq:prob_dist}. Therefore, we have (conditional on $G_{t-1}$), that $(C^{(t)}_0,\ldots,C^{(t)}_{t-1}) \sim \MultinomialDist(M; P^{(t)}_0, P^{(t)}_1, \ldots, P^{(t)}_{t-1})$, where $C_i^{(t)}$ denote the number of times the vertex~$v_i$ is chosen randomly as a neighbour of~$v_t$ at time~$t$. Therefore,

\[
		\Prob[\big]{ (C^{(t)}_0,\ldots,C^{(t)}_{t-1}) = (c_0, \ldots c_{t-1}) \given G_{t-1}} 
		\wrel= \binom{M}{c_{0}, \ldots ,  c_{t-1}} (P^{(t)}_0)^{c_{0}} \ldots(P^{(t)}_{t-1})^{c_{t-1}} .
\]

\subsection[Proof of Theorem \protect\ref*{thm:lower-bound-lbl}]{Proof of \wref{thm:lower-bound-lbl}}
\label{sec:proof-lower-bound-lbl}
In this subsection, we derive a lower bound on $\lg(1/\Prob{G})$ for $\Prob{G}$ the probability that a given graph $G\sim\PrefAttch$ results from a preferential-attachment process $\PrefAttch$. 

Recall that~$G = G_{n}$. Using the product rule, we get 
\[
		\Prob{G} 
	\wwrel= 
		\Prob{\OutNeighborhood(v_{n}) \given G_{n-1}} 
		\cdot
		\Prob{\OutNeighborhood(v_{n-1}) \given G_{n-2}} 
		\;\cdots\;
		\Prob{\OutNeighborhood(v_{2}) \given G_{1}} 
		\cdot \Prob{G_1}.
\]
This implies
\begin{align*}
		\lg\left(\frac{1}{\Prob{G}}\right) 
	&\wwrel= 
		\sum_{t = 2}^{n} \lg \left(\frac{1}{ \Prob{\OutNeighborhood(v_t)\given G_{t-1}} } \right) 
	\\&\wwrel= 
		\sum_{t = 2}^{n} \left[\sum_{i= 0}^{t-1} C_{i}^{(t)} \lg \left(\frac{2(t-1)M}{d_{t-1}(v_i)}\right) 
		\bin- \lg\left(\binom{M}{C_{0}^{(t)},\ldots, C_{t-1}^{(t)}}\right)  \right].
\numberthis\label{eq:lower_bound}
\end{align*}
Consider $- \sum_{t = 2}^{n}  \lg\left(\binom{M}{C_{0}^{(t)}, \ldots, C_{t-1}^{(t)}}\right)$. We have
\begin{align*}
		-\lg\left(\binom{M}{C_{0}^{(t)},\ldots, C_{t-1}^{(t)}}\right)
	&\wwrel=  
		\lg\left(\frac{C_{0}^{(t)}! \ldots C_{t-1}^{(t)}!}{M!}\right) 
	\wwrel\geq 
		\lg\Bigl(\frac{1}{M!}\Bigr) 
	\wwrel= 
		- \lg(M!).
\end{align*}

\noindent Now let us consider 
the first summand in \wref{eq:lower_bound},
\begin{align*}
	&
		\sum_{t = 2}^{n} \sum_{i= 0}^{t-1} C_{i}^{(t)} \lg \left(\frac{2(t-1)M}{d_{t-1}(v_i)}\right)
\\	&\wwrel=  
		\sum_{t = 1}^{n-1}\left(\sum_{i=0}^{t} C_{i}^{(t+1)} \lg \left(2tM\right)\right) 
		\bin-  \sum_{t = 1}^{n-1}\left(\sum_{i= 0}^{t} C_{i}^{(t+1)} \lg \left({d_{t}(v_i)}\right)\right).
\numberthis\label{eq:lb-first-summand-split}
\end{align*}
Let us analyse the first summand. With $\sum_{i= 0}^{t} C_{i}^{(t+1)} = M$, we find
\begin{align*}
		\sum_{t = 1}^{n-1}\left(\sum_{i=0}^{t} C_{i}^{(t+1)} \lg \left(2tM\right)\right) 
	&\wwrel=
		\sum_{t = 1}^{n-1}\lg \left(2tM\right) \left(\sum_{i= 0}^{t} C_{i}^{(t+1)} \right)
\\	&\wwrel= 
		M (n-1)\lg(2M) + M\lg((n-1)!).
\end{align*}

\noindent Now consider the second summand in \wref{eq:lb-first-summand-split}.
Swapping the order of summation yields
\begin{align*}
		-\sum_{t = 1}^{n-1}\left(\sum_{i= 0}^{t} C_{i}^{(t+1)} \lg \left({d_{t}(v_i)}\right)\right)
	&\wwrel=
		-\sum_{i = 0}^{n-1}\left(\sum_{t=i+1}^{n-1} C_{i}^{(t+1)} \lg \left({d_{t}(v_i)}\right)\right).
\end{align*}

Let~$T_i = \{t\ge 2: C_{i}^{(t)} > 0\}$ denote the set of timestamps such that the vertex~$v_i$ is selected as an out-neighbour; we have $T_i = \{t^{(i)}_1,\ldots,t^{(i)}_{\InDegreeInGraph(v_i)}\}$ with $t^{(i)}_1 \le\cdots\le t^{(i)}_{\InDegreeInGraph(v_i)}$, \ie, just before time $t_k^{(i)}$, $v_i$'s total degree was $d_{t_k^{(i)}-1}(v_i) = M+k-1$.%
\footnote{%
	For ease of notation, we do not count the $M$ edges from $v_1$ to $v_0$ in $\InDegreeInGraph(v_0)$ here;
	then for all vertices, we have $\DegreeInGraph(v_i) = \InDegreeInGraph(v_i)+M$.
	That is also why we only include times $t\ge 2$ in $T_i$.
}
So far, the analysis works for general graphs;
for the following simplification, we assume that no parallel edges%
\footnote{
    We note that $G_M$ necessarily contains parallel edges, and our assumptions concern the edges chosen for $t>M$.
    The contribution of these parallel edges in the calculation below can be shown to be $O(M)$ overall,
    so we ignore their presence for legibility.
}
are added in the graph (for $t\ge 2$).
Then $C_{i}^{(t)} \le 1$ and the summand becomes
\begin{align*}
		-\sum_{i = 0}^{n-1}\sum_{\substack{t=i+1\\t\ge 2}}^{n} C_{i}^{(t)} \lg \left({d_{t-1}(v_i)}\right)
	&\wwrel=
		-\sum_{i = 0}^{n-1}\sum_{k=1}^{\InDegreeInGraph(v_i)} \lg \left(M+k-1\right)
\\	&\wwrel=
		-\sum_{i = 0}^{n-1}\left(\sum_{k=1}^{\InDegreeInGraph(v_i)+M-1} \lg(k) 
			-  \sum_{k=1}^{M-1} \lg(k)\right)
\\	&\wwrel=
		-\sum_{i=0}^{n-1} \lg( (\DegreeInGraph(v_i)-1)!)
		\bin + n\lg((M-1)!).
\end{align*}
Therefore, Equation~\eqref{eq:lower_bound} becomes
\begin{align*}
		\lg\left(\frac{1}{\mathbb{P}[G]}\right) 
	&\wwrel\ge
		M(n-1)\lg(2M) + M(\lg((n-1)!))
\\*	&\wwrel\ppe{}
		-\sum_{i=0}^{n-1} \lg( (\DegreeInGraph(v_i)-1)!)
					\bin + n\lg((M-1)!)
\\*	&\wwrel\ppe{}
		- (n-1)\lg(M!)
\\	&\wwrel=
		M n \lg (2Mn) %
		\bin-\sum_{i=0}^{n-1} \DegreeInGraph(v_i)\lg( \DegreeInGraph(v_i)) %
		\wbin\pm O(Mn)
\shortintertext{(using $\DegreeInGraph(v) = \InDegreeInGraph(v) + M$)}
	&\wwrel=
		Mn \underbrace{\sum_{i=0}^{n-1} \frac{\InDegreeInGraph(v_i)}{Mn} \lg \left(\frac{2Mn}{\DegreeInGraph(v_i)}\right)}_{(*)}
		\bin-\underbrace{\sum_{i=0}^{n-1} M\lg( \DegreeInGraph(v_i))}_{(\dag)}
		\wbin\pm O(Mn)
\intertext{(bounding $(*)$ by Gibbs' inequality, and $(\dag)$ by the log-sum inequality)}
	&\wwrel\ge
		Mn \sum_{i=0}^{n-1} \frac{\InDegreeInGraph(v_i)}{Mn} \lg \left(\frac{Mn}{\InDegreeInGraph(v_i)}\right)
		\wbin\pm O(Mn\lg M)
\\	&\wwrel=
		\DegreeEntropy(G)
		\wbin\pm O(Mn\lg M).
\numberthis\label{eq:lb-entropy}
\end{align*}
This proves \wref{thm:lower-bound-lbl}.

\ifproceedings{}{
	\medskip\noindent
	(The proof of \wref{thm:lower-bound-ulbl} is given in \wref{app:proof-lower-bound-ulbl}.)
}

\section{Data Structures}
\label{sec:data-structures}

In this section, we describe data structures that can represent a directed graph generated by the preferential-attachment process in compressed form, while allowing for efficient navigational queries without decompression.

For labelled graphs, we only use the wavelet tree portion of the data structure described in the following. 
We omit the simple modifications here and present the data structure for unlabelled graphs in detail;  
(see also~\cite[\S5.3]{navarro2014wavelet} for the labelled case).

\begin{figure}[tbp]
    \centering
    \ifsubmission{
        \scalebox{.7}{% !TeX root = pref-attch-main
\begin{tikzpicture}[
    node distance = 0pt,
    square_blue/.style = {draw=blue!60, fill=blue!5, very thick, minimum height=1.2em, minimum width=2em, outer sep=0pt},
    square_green/.style = {draw=green!60, fill=green!5, very thick, minimum height=1.2em, minimum width=2em, outer sep=0pt}, scale=1.2,
    ]

    % G
    
    \node[thick,shape=circle,rounded corners, draw=black, minimum size=6mm] (0) at (-1.25,0) {$v_0$};
    \node[thick,shape=circle,rounded corners, draw=black, minimum size=6mm] (1) at (0,0)  {$v_1$};
    \node[thick,shape=circle,rounded corners, draw=black, minimum size=6mm] (2) at (2,0)  {$v_2$};
    \node[thick,shape=circle,rounded corners, draw=black, minimum size=6mm] (3) at (4,0)  {$v_3$};
    \node[thick,shape=circle,rounded corners, draw=black, minimum size=6mm] (4) at (6,0)  {$v_4$};
    \node[thick,shape=circle,rounded corners, draw=black, minimum size=6mm] (5) at (8,0)  {$v_5$};
    \draw[>=Stealth, ->,very thick] (1) edge[densely dashed,red]  (0);
    \draw[>=Stealth, ->,very thick] (1) edge[bend right=30]  (0);
    \draw[>=Stealth, ->,very thick] (1) edge[bend left=30]  (0);
    \draw[>=Stealth, ->,very thick] (2) edge[bend right=20] (1);
    \draw[>=Stealth, ->,very thick] (2) edge[bend left=20] (1);
    \draw[>=Stealth, ->,very thick] (3) edge[bend left=20] (1);
    \draw[>=Stealth, ->,very thick] (3) edge[bend right=35] (1);
    \draw[>=Stealth, ->,very thick] (4) edge[bend right=20] (2);
    \draw[>=Stealth, ->,very thick] (4) edge[bend left=20] (2);
    \draw[>=Stealth, ->,very thick] (4) edge[red, densely dashed] (3);
    \draw[>=Stealth, ->,very thick] (5) edge[bend right=20] (4);
    \draw[>=Stealth, ->,very thick] (5) edge[bend left=20] (4);
    \draw[>=Stealth, ->,very thick] (5) edge[bend right=35,densely dashed,red] (3);
    \draw[>=Stealth, ->,very thick] (2) edge[red, densely dashed] (1);
    \draw[>=Stealth, ->,very thick] (3) edge[bend right=20,red, densely dashed] (1);

    \node[] (G) at (-2,0){$G$:};
    
    \node[] (N1) at (0,-1){$\OutNeighborhood(1)$};
    \node[] (N1) at (2,-1){$\OutNeighborhood(2)$};
    \node[] (N1) at (4,-1){$\OutNeighborhood(3)$};
    \node[] (N1) at (6,-1){$\OutNeighborhood(4)$};
    \node[] (N1) at (8,-1){$\OutNeighborhood(5)$};
    
    % A
    \node[] (A) at (-2,-1.5){$A$:};
    
    \node[square_blue]   (1b) at (0,-1.5){0};
    \node[square_blue, left=of 1b] (1a)  {0};
    \node[square_blue, right=of 1b] (1c) {0};

    \node[square_green]   (2b) at (2,-1.5){1};
    \node[square_green, left=of 2b] (2a) {1};
    \node[square_green, right=of 2b] (2c) {1};

    \node[square_blue]   (3b) at (4,-1.5){1};
    \node[square_blue, left=of 3b] (3a) {1};
    \node[square_blue, right=of 3b] (3c) {1};

    \node[square_green]   (4b) at (6,-1.5){2};
    \node[square_green, left=of 4b] (4a) {3};
    \node[square_green, right=of 4b] (4c) {2};

    \node[square_blue]   (5b) at (8,-1.5){4};
    \node[square_blue, left=of 5b] (5a) {3};
    \node[square_blue, right=of 5b] (5c) {4};

    % A'
    \node[] (Ap) at (-2.05,-2){$A'$:};

    \node[square_blue]   (1bp) at (0,-2){0};
    %\node[square_blue, left=of 1bp] (1ap) {1};
    \node[square_blue, right=of 1bp] (1cp) {0};

    \node[square_green]   (2bp) at (2,-2){1};
    %\node[square_green, left=of 2bp] (2ap) {1};
    \node[square_green, right=of 2bp] (2cp) {1};

    \node[square_blue]   (3bp) at (4,-2){1};
    \node[square_blue, left=of 3bp] (3ap) {1};
    %\node[square_blue, right=of 3bp] (3cp) {1};

    \node[square_green]   (4bp) at (6,-2){2};
    %\node[square_green, left=of 4bp] (4ap) {1};
    \node[square_green, right=of 4bp] (4cp) {2};

    \node[square_blue]   (5bp) at (8,-2){4};
    %\node[square_blue, left=of 5bp] (5ap) {1};
    \node[square_blue, right=of 5bp] (5cp) {4};

    % T

    \node[] (G) at (-2,-4.5){$T$:};

    \node[shape=circle,rounded corners, draw=black, minimum size=6mm] (0t) at (4,-3){$0$};
    \node[shape=circle,rounded corners, draw=black, minimum size=6mm] (1t) at (4,-4){$1$};
    \node[shape=circle,rounded corners, draw=black, minimum size=6mm] (2t) at (3,-5){$2$};
    \node[shape=circle,rounded corners, draw=black, minimum size=6mm] (3t) at (5,-5){$3$};
    \node[shape=circle,rounded corners, draw=black, minimum size=6mm] (4t) at (4,-6){$4$};
    \node[shape=circle,rounded corners, draw=black, minimum size=6mm] (5t) at (6,-6){$5$};
    \draw[>=Triangle, -, ultra thick] (1t) edge[red, densely dashed] coordinate (1e) (0t);
    \draw[>=Triangle, -, ultra thick] (1t) edge[red, densely dashed] coordinate (2e) (2t);
    \draw[>=Triangle, -, ultra thick] (1t) edge[red, densely dashed] coordinate (3e) (3t);
    \draw[>=Triangle, -, ultra thick] (3t) edge[red, densely dashed] coordinate (4e) (4t);
    \draw[>=Triangle, -, ultra thick] (3t) edge[red, densely dashed] coordinate (5e) (5t);
    \draw[->, thick, shorten >=3pt] (1a) edge[gray, densely dotted, out=-90,in=180] (1e);
    \draw[->, thick, shorten >=3pt] (2a) edge[gray, densely dotted, out=-90,in=135] (2e);
    \draw[->, thick, shorten >=3pt] (3c) edge[gray, densely dotted, out=-90,in=45] (3e);
    \draw[->, thick, shorten >=3pt] (4a) edge[gray, densely dotted, out=-90,in=125] (4e);
    \draw[->, thick, shorten >=3pt] (5a) edge[gray, densely dotted, out=-90,in=45] (5e);
    
\end{tikzpicture}}
    }{
	    \ifproceedings{
	        \resizebox{\linewidth}!{% !TeX root = pref-attch-main
\begin{tikzpicture}[
    node distance = 0pt,
    square_blue/.style = {draw=blue!60, fill=blue!5, very thick, minimum height=1.2em, minimum width=2em, outer sep=0pt},
    square_green/.style = {draw=green!60, fill=green!5, very thick, minimum height=1.2em, minimum width=2em, outer sep=0pt}, scale=1.2,
    ]

    % G
    
    \node[thick,shape=circle,rounded corners, draw=black, minimum size=6mm] (0) at (-1.25,0) {$v_0$};
    \node[thick,shape=circle,rounded corners, draw=black, minimum size=6mm] (1) at (0,0)  {$v_1$};
    \node[thick,shape=circle,rounded corners, draw=black, minimum size=6mm] (2) at (2,0)  {$v_2$};
    \node[thick,shape=circle,rounded corners, draw=black, minimum size=6mm] (3) at (4,0)  {$v_3$};
    \node[thick,shape=circle,rounded corners, draw=black, minimum size=6mm] (4) at (6,0)  {$v_4$};
    \node[thick,shape=circle,rounded corners, draw=black, minimum size=6mm] (5) at (8,0)  {$v_5$};
    \draw[>=Stealth, ->,very thick] (1) edge[densely dashed,red]  (0);
    \draw[>=Stealth, ->,very thick] (1) edge[bend right=30]  (0);
    \draw[>=Stealth, ->,very thick] (1) edge[bend left=30]  (0);
    \draw[>=Stealth, ->,very thick] (2) edge[bend right=20] (1);
    \draw[>=Stealth, ->,very thick] (2) edge[bend left=20] (1);
    \draw[>=Stealth, ->,very thick] (3) edge[bend left=20] (1);
    \draw[>=Stealth, ->,very thick] (3) edge[bend right=35] (1);
    \draw[>=Stealth, ->,very thick] (4) edge[bend right=20] (2);
    \draw[>=Stealth, ->,very thick] (4) edge[bend left=20] (2);
    \draw[>=Stealth, ->,very thick] (4) edge[red, densely dashed] (3);
    \draw[>=Stealth, ->,very thick] (5) edge[bend right=20] (4);
    \draw[>=Stealth, ->,very thick] (5) edge[bend left=20] (4);
    \draw[>=Stealth, ->,very thick] (5) edge[bend right=35,densely dashed,red] (3);
    \draw[>=Stealth, ->,very thick] (2) edge[red, densely dashed] (1);
    \draw[>=Stealth, ->,very thick] (3) edge[bend right=20,red, densely dashed] (1);

    \node[] (G) at (-2,0){$G$:};
    
    \node[] (N1) at (0,-1){$\OutNeighborhood(1)$};
    \node[] (N1) at (2,-1){$\OutNeighborhood(2)$};
    \node[] (N1) at (4,-1){$\OutNeighborhood(3)$};
    \node[] (N1) at (6,-1){$\OutNeighborhood(4)$};
    \node[] (N1) at (8,-1){$\OutNeighborhood(5)$};
    
    % A
    \node[] (A) at (-2,-1.5){$A$:};
    
    \node[square_blue]   (1b) at (0,-1.5){0};
    \node[square_blue, left=of 1b] (1a)  {0};
    \node[square_blue, right=of 1b] (1c) {0};

    \node[square_green]   (2b) at (2,-1.5){1};
    \node[square_green, left=of 2b] (2a) {1};
    \node[square_green, right=of 2b] (2c) {1};

    \node[square_blue]   (3b) at (4,-1.5){1};
    \node[square_blue, left=of 3b] (3a) {1};
    \node[square_blue, right=of 3b] (3c) {1};

    \node[square_green]   (4b) at (6,-1.5){2};
    \node[square_green, left=of 4b] (4a) {3};
    \node[square_green, right=of 4b] (4c) {2};

    \node[square_blue]   (5b) at (8,-1.5){4};
    \node[square_blue, left=of 5b] (5a) {3};
    \node[square_blue, right=of 5b] (5c) {4};

    % A'
    \node[] (Ap) at (-2.05,-2){$A'$:};

    \node[square_blue]   (1bp) at (0,-2){0};
    %\node[square_blue, left=of 1bp] (1ap) {1};
    \node[square_blue, right=of 1bp] (1cp) {0};

    \node[square_green]   (2bp) at (2,-2){1};
    %\node[square_green, left=of 2bp] (2ap) {1};
    \node[square_green, right=of 2bp] (2cp) {1};

    \node[square_blue]   (3bp) at (4,-2){1};
    \node[square_blue, left=of 3bp] (3ap) {1};
    %\node[square_blue, right=of 3bp] (3cp) {1};

    \node[square_green]   (4bp) at (6,-2){2};
    %\node[square_green, left=of 4bp] (4ap) {1};
    \node[square_green, right=of 4bp] (4cp) {2};

    \node[square_blue]   (5bp) at (8,-2){4};
    %\node[square_blue, left=of 5bp] (5ap) {1};
    \node[square_blue, right=of 5bp] (5cp) {4};

    % T

    \node[] (G) at (-2,-4.5){$T$:};

    \node[shape=circle,rounded corners, draw=black, minimum size=6mm] (0t) at (4,-3){$0$};
    \node[shape=circle,rounded corners, draw=black, minimum size=6mm] (1t) at (4,-4){$1$};
    \node[shape=circle,rounded corners, draw=black, minimum size=6mm] (2t) at (3,-5){$2$};
    \node[shape=circle,rounded corners, draw=black, minimum size=6mm] (3t) at (5,-5){$3$};
    \node[shape=circle,rounded corners, draw=black, minimum size=6mm] (4t) at (4,-6){$4$};
    \node[shape=circle,rounded corners, draw=black, minimum size=6mm] (5t) at (6,-6){$5$};
    \draw[>=Triangle, -, ultra thick] (1t) edge[red, densely dashed] coordinate (1e) (0t);
    \draw[>=Triangle, -, ultra thick] (1t) edge[red, densely dashed] coordinate (2e) (2t);
    \draw[>=Triangle, -, ultra thick] (1t) edge[red, densely dashed] coordinate (3e) (3t);
    \draw[>=Triangle, -, ultra thick] (3t) edge[red, densely dashed] coordinate (4e) (4t);
    \draw[>=Triangle, -, ultra thick] (3t) edge[red, densely dashed] coordinate (5e) (5t);
    \draw[->, thick, shorten >=3pt] (1a) edge[gray, densely dotted, out=-90,in=180] (1e);
    \draw[->, thick, shorten >=3pt] (2a) edge[gray, densely dotted, out=-90,in=135] (2e);
    \draw[->, thick, shorten >=3pt] (3c) edge[gray, densely dotted, out=-90,in=45] (3e);
    \draw[->, thick, shorten >=3pt] (4a) edge[gray, densely dotted, out=-90,in=125] (4e);
    \draw[->, thick, shorten >=3pt] (5a) edge[gray, densely dotted, out=-90,in=45] (5e);
    
\end{tikzpicture}}
	    }{
	        % !TeX root = pref-attch-main
\begin{tikzpicture}[
    node distance = 0pt,
    square_blue/.style = {draw=blue!60, fill=blue!5, very thick, minimum height=1.2em, minimum width=2em, outer sep=0pt},
    square_green/.style = {draw=green!60, fill=green!5, very thick, minimum height=1.2em, minimum width=2em, outer sep=0pt}, scale=1.2,
    ]

    % G
    
    \node[thick,shape=circle,rounded corners, draw=black, minimum size=6mm] (0) at (-1.25,0) {$v_0$};
    \node[thick,shape=circle,rounded corners, draw=black, minimum size=6mm] (1) at (0,0)  {$v_1$};
    \node[thick,shape=circle,rounded corners, draw=black, minimum size=6mm] (2) at (2,0)  {$v_2$};
    \node[thick,shape=circle,rounded corners, draw=black, minimum size=6mm] (3) at (4,0)  {$v_3$};
    \node[thick,shape=circle,rounded corners, draw=black, minimum size=6mm] (4) at (6,0)  {$v_4$};
    \node[thick,shape=circle,rounded corners, draw=black, minimum size=6mm] (5) at (8,0)  {$v_5$};
    \draw[>=Stealth, ->,very thick] (1) edge[densely dashed,red]  (0);
    \draw[>=Stealth, ->,very thick] (1) edge[bend right=30]  (0);
    \draw[>=Stealth, ->,very thick] (1) edge[bend left=30]  (0);
    \draw[>=Stealth, ->,very thick] (2) edge[bend right=20] (1);
    \draw[>=Stealth, ->,very thick] (2) edge[bend left=20] (1);
    \draw[>=Stealth, ->,very thick] (3) edge[bend left=20] (1);
    \draw[>=Stealth, ->,very thick] (3) edge[bend right=35] (1);
    \draw[>=Stealth, ->,very thick] (4) edge[bend right=20] (2);
    \draw[>=Stealth, ->,very thick] (4) edge[bend left=20] (2);
    \draw[>=Stealth, ->,very thick] (4) edge[red, densely dashed] (3);
    \draw[>=Stealth, ->,very thick] (5) edge[bend right=20] (4);
    \draw[>=Stealth, ->,very thick] (5) edge[bend left=20] (4);
    \draw[>=Stealth, ->,very thick] (5) edge[bend right=35,densely dashed,red] (3);
    \draw[>=Stealth, ->,very thick] (2) edge[red, densely dashed] (1);
    \draw[>=Stealth, ->,very thick] (3) edge[bend right=20,red, densely dashed] (1);

    \node[] (G) at (-2,0){$G$:};
    
    \node[] (N1) at (0,-1){$\OutNeighborhood(1)$};
    \node[] (N1) at (2,-1){$\OutNeighborhood(2)$};
    \node[] (N1) at (4,-1){$\OutNeighborhood(3)$};
    \node[] (N1) at (6,-1){$\OutNeighborhood(4)$};
    \node[] (N1) at (8,-1){$\OutNeighborhood(5)$};
    
    % A
    \node[] (A) at (-2,-1.5){$A$:};
    
    \node[square_blue]   (1b) at (0,-1.5){0};
    \node[square_blue, left=of 1b] (1a)  {0};
    \node[square_blue, right=of 1b] (1c) {0};

    \node[square_green]   (2b) at (2,-1.5){1};
    \node[square_green, left=of 2b] (2a) {1};
    \node[square_green, right=of 2b] (2c) {1};

    \node[square_blue]   (3b) at (4,-1.5){1};
    \node[square_blue, left=of 3b] (3a) {1};
    \node[square_blue, right=of 3b] (3c) {1};

    \node[square_green]   (4b) at (6,-1.5){2};
    \node[square_green, left=of 4b] (4a) {3};
    \node[square_green, right=of 4b] (4c) {2};

    \node[square_blue]   (5b) at (8,-1.5){4};
    \node[square_blue, left=of 5b] (5a) {3};
    \node[square_blue, right=of 5b] (5c) {4};

    % A'
    \node[] (Ap) at (-2.05,-2){$A'$:};

    \node[square_blue]   (1bp) at (0,-2){0};
    %\node[square_blue, left=of 1bp] (1ap) {1};
    \node[square_blue, right=of 1bp] (1cp) {0};

    \node[square_green]   (2bp) at (2,-2){1};
    %\node[square_green, left=of 2bp] (2ap) {1};
    \node[square_green, right=of 2bp] (2cp) {1};

    \node[square_blue]   (3bp) at (4,-2){1};
    \node[square_blue, left=of 3bp] (3ap) {1};
    %\node[square_blue, right=of 3bp] (3cp) {1};

    \node[square_green]   (4bp) at (6,-2){2};
    %\node[square_green, left=of 4bp] (4ap) {1};
    \node[square_green, right=of 4bp] (4cp) {2};

    \node[square_blue]   (5bp) at (8,-2){4};
    %\node[square_blue, left=of 5bp] (5ap) {1};
    \node[square_blue, right=of 5bp] (5cp) {4};

    % T

    \node[] (G) at (-2,-4.5){$T$:};

    \node[shape=circle,rounded corners, draw=black, minimum size=6mm] (0t) at (4,-3){$0$};
    \node[shape=circle,rounded corners, draw=black, minimum size=6mm] (1t) at (4,-4){$1$};
    \node[shape=circle,rounded corners, draw=black, minimum size=6mm] (2t) at (3,-5){$2$};
    \node[shape=circle,rounded corners, draw=black, minimum size=6mm] (3t) at (5,-5){$3$};
    \node[shape=circle,rounded corners, draw=black, minimum size=6mm] (4t) at (4,-6){$4$};
    \node[shape=circle,rounded corners, draw=black, minimum size=6mm] (5t) at (6,-6){$5$};
    \draw[>=Triangle, -, ultra thick] (1t) edge[red, densely dashed] coordinate (1e) (0t);
    \draw[>=Triangle, -, ultra thick] (1t) edge[red, densely dashed] coordinate (2e) (2t);
    \draw[>=Triangle, -, ultra thick] (1t) edge[red, densely dashed] coordinate (3e) (3t);
    \draw[>=Triangle, -, ultra thick] (3t) edge[red, densely dashed] coordinate (4e) (4t);
    \draw[>=Triangle, -, ultra thick] (3t) edge[red, densely dashed] coordinate (5e) (5t);
    \draw[->, thick, shorten >=3pt] (1a) edge[gray, densely dotted, out=-90,in=180] (1e);
    \draw[->, thick, shorten >=3pt] (2a) edge[gray, densely dotted, out=-90,in=135] (2e);
    \draw[->, thick, shorten >=3pt] (3c) edge[gray, densely dotted, out=-90,in=45] (3e);
    \draw[->, thick, shorten >=3pt] (4a) edge[gray, densely dotted, out=-90,in=125] (4e);
    \draw[->, thick, shorten >=3pt] (5a) edge[gray, densely dotted, out=-90,in=45] (5e);
    
\end{tikzpicture}
    	}
    }
    \caption{Illustration of the construction of $T$, $A$ and $A'$ on a sample preferential-attachment graph where $n=5$ and $M=3$. The dashed red edges in $G$ represent the edges that are represented in $T$. Note that the vertices are labelled by their preorder index in $T$.}
    \label{fig:ds-construction}
\end{figure}
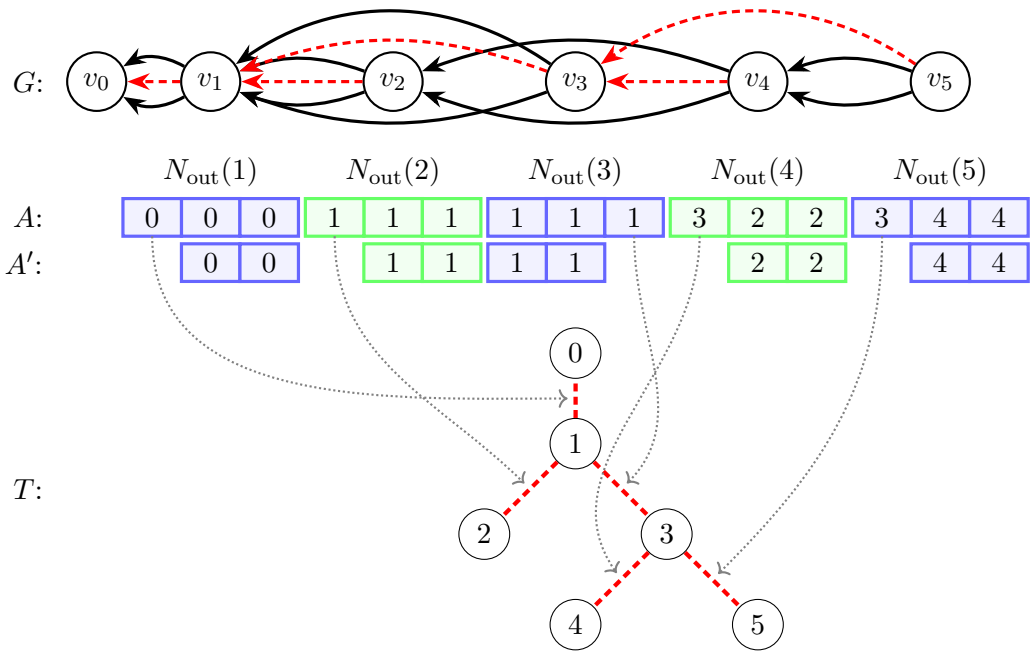

\subsection{Construction}
\label{ss:construction}

Let $G' \sim \PrefAttch$ be generated by the Barabási-Albert model.
By iteratively ``peeling off'' degree-$M$ vertices, we uniquely obtain the directed version $G=\DAGof{G'}$~\cite{LuczakMagnerSzpankowski2019}.

We now sort vertices by indegree and construct a tree $T$ informally by choosing, for each vertex $v_t$ in $G$ (except for $v_0$), a parent among its out-neighbours with minimal (in-)degree.
Formally, define a bijective \emph{rank} function $\sigma : \{v_0,\ldots,v_n\} \rightarrow [0..n]$ such that $\InDegreeInGraph(u) > \InDegreeInGraph(v)$ implies $\sigma(u) > \sigma(v)$. 
Construct a tree $T$ with $V(T)=V(G)$ and
\[
		E(T)
	\wwrel=
		\bigcup_{t=1}^{n} \Bigl\{(v_t,\mathop{\arg \min}\limits_{w\in \OutNeighborhood(v_t)} \sigma(w))\Bigr\}.
\]
Observe that $T$ is indeed a tree since it consists of $n+1$ vertices and $n$ edges and is acyclic by definition. Vertex $v_0$ is the root of $T$.

We will store $T$ using \wref{lem:tree-succinct}, where operations refer to nodes via their preorder index.
We hence perform a preorder traversal of $T$ from its root, and relabel $v_t\in V(G)$ by its index in the traversal, starting from $0$. 

For each vertex $v\neq 0$ in $G$, define an order on its outgoing edges such that the edge $(v,u)$, for some $u$, is the first if and only if $(v,u)\in E(T)$, and all other vertices $(v,w)$, $u \neq w$, are ordered arbitrarily. Finally, construct a static array $A'$ such that the $(i+1)$th ($i \in [1..M)$) out-neighbour of vertex $j$ in $G$ is stored in $A'[(j-1)(M-1)+i]$.
We store $A'$ in a wavelet tree using \wref{lem:wavelet-tree-simple}.

\subsection{Operations}
\label{sec:operations}

Building upon the aforementioned data structures, we present algorithms to efficiently handle a variety of navigational queries on $G$. 
Define $\mathsf{source}(i)=\lceil\frac{i}{M-1}\rceil$ as the source of the edge whose target is stored in $A'[i]$.

\paragraph{Computing the $i$th out-neighbour of $v$} 
The $i$th \emph{out-neighbour} of $v$ in $G$ can be determined as follows: if $i = 1$, it is the parent of $v$ in $T$; otherwise, it is given by $A'[(v-1)(M-1)+i-1]$ (the $(i-1)$th element of $A'[(v-1)(M-1)+1 .. v(M-1)]$, which represents all but the first out-neighbour of $v$).

\paragraph{Computing the $i$th in-neighbour of $v$} 
To determine the $i$th \emph{in-neighbour} of $v$, we distinguish two cases:  
\begin{enumerate}
\item If $v$ has at least $i$ children in $T$, then its $i$th in-neighbour in $G$ is simply its $i$th child in~$T$.  
\item Otherwise, if $i > \mathsf{degree}(T, v)$, we determine the $(i - \mathsf{degree}(T,v))$th occurrence of $v$ in $A'$. The vertex whose neighbourhood includes this occurrence is found by computing $\mathsf{source}(\mathsf{select}_v(S,i-\mathsf{degree}(T,v)))$.
\end{enumerate}
Using these two operations, we can also iterate through the neighbourhoods or return all (in- resp.\ out-) neighbours.

\paragraph{Computing the in-degree of $v$} 
The in-degree of a vertex $v$ can be determined by summing the number of its children in $T$ (computable by $\mathsf{degree}(T,v)$) and its occurrences in $A'$ (computable by $\mathsf{rank}_v(A',n(M-1)$).

\paragraph{Determining if $u$ and $v$ are adjacent} 
If $u$ and $v$ are adjacent, either (1) $u$ is the parent of $v$ in $T$, $\mathsf{parent}(T,v)=u$, (2) vice versa, $\mathsf{parent}(T,u) = v$, or 
(3) $v$ occurs in $A'[(u-1)(M-1)+1..u(M-1)]$ or (4) vice versa, $u$ occurs in $A'[(v-1)(M-1)+1..v(M-1)]$.
(If either of them is $v_0$, two cases need not be checked.)
Conditions (1) and (2) are directly supported on $T$; for (3) and (4), we can use rank:
If $\mathsf{rank}_v(A',u(M-1))-\mathsf{rank}_v(A',(u-1)(M-1)) \geq 1$, then $v$ occurs at least once, so $v$ and $u$ are adjacent. (4) is similar.

\medskip\noindent
To conclude the proof of \wref{thm:data-structure}, 
it remains to analyse the space usage of our data structure;
\ifsubmission{%
    we give the details in \wref{app:space-analysis-ds}.
}{%
	\ifproceedings{%
		details are given in the appendix of the \arxivpaper.
	}{%
		details are given in \wref{app:space-analysis-ds}.
	}
}%

\begin{remark}[General graphs]
    We point out that the data-structure techniques above
    can be extended to general graphs. 
    Using a compressed bit vector to mark where the next neighbourhoods begin,
    we can support the same queries with $n \lg(m/n) + O(n)$ extra space.
\end{remark}

%
%
%
%
%
%
%
%
%
%
%
%
%
%
%
%
%
%
%
%
%
%
%
%
%
%
%
%

% !TeX root = pref-attch-main

\section{Conclusion}
We designed a compressed representation for graphs generated by the Barabási-Albert model of random preferential-attachment graphs approaching the instance-optimal $\lg(1/\Prob{G})$ bits of space, where~$\Prob{G}$ is the probability for the graph to arise in the model. We further related $\lg(1/\Prob{G})$ to the empirical degree-entropy of $G$, $\DegreeEntropy(G)$; they coincide up to an error term in both labelled and unlabelled graphs. Our compressed representation supports navigational queries in $O(\log n)$ time and can simulate access to an adjacency-list representation. 

Future work might study other models of preferential attachment. For example, the model introduced by Cooper and Frieze~\cite{cooper2003general}, where the number of edges added follows a given distribution. One can also consider a non-linear model in which the probability of a vertex, when $M$ edges are added, being chosen is proportional to its degree raised to some power $\alpha$. (Here $\alpha = 1$.)

More broadly speaking, this work merely made initial strides into data structures that adapt to information-theoretic measures of compressibility in graphs, leaving many avenues for future work open.
For example, unlike in trees, a convincing notion of higher-order entropy for graphs is not yet emerging, let alone data structures approaching them.

	\myacknowledgements
%

%
% !BIB program = bibtex
\bibliography{references}

\clearpage
\appendix
\ifkoma{\addpart{Appendix}}{}

%
%
%
%
%
%\ifarxiv{
    \section[Proof of Theorem \protect\ref*{thm:lower-bound-ulbl} (Lower Bound for Unlabelled PA Graphs)]{Proof of \protect\wref{thm:lower-bound-ulbl} (Lower Bound for Unlabelled PA Graphs)}
\label{app:proof-lower-bound-ulbl}

Assuming an undirected graph $G$ has been generated by the Barabási-Albert process,
we can uniquely reconstruct its directed version by iteratively ``peeling'' off degree-$M$ vertices~\cite{LuczakMagnerSzpankowski2019}.
Note that the same is \emph{not} true about the arrival times; in general the resulting DAG contains only partial information about the vertex arrival order: any linear extension of the partial order induced by the DAG is an admissible arrival order. 
Our analysis builds on results by \L{}uczak et al. who proved that 
(1) most PA graphs have $n!\cdot2^{-O(n \lg \lg n)}$ admissible arrival orders~\cite[p.\,715]{LuczakMagnerSzpankowski2019} and 
(2) that each of these arises from $\PrefAttch$ with equal probability~\cite[Lem.\,7]{LuczakMagnerSzpankowski2019}.
We reproduce the relevant arguments here.

We define the admissible set $\mathrm{Adm}(S)$ of a given unlabelled graph $S$ to be the set of all labelled graphs $G$ with $S(G) = S$ such that $G$ could have been generated according to the preferential-attachment model. We can also define~$\mathrm{Adm}(G) := \mathrm{Adm}(S(G))$.
For a graph $G$, we define $\Gamma(G)$ to be the set of permutations~$\pi$ such that $\pi(G) \in \mathrm{Adm}(G)$. 
Let~$\mathrm{Aut}(G)$ be the automorphism group of~$G$.
For any~$G$, we have the following equation~\cite{magner2017recovery}
\begin{equation}
    |\mathrm{Adm}(G)| := \frac{|\Gamma(G)|}{|\mathrm{Aut}(G)|}.
\end{equation}

For a graph~$G \sim \PrefAttch$ and $\UDAGof{G}$, every possible way to order the vertices will result in the \emph{same} probability for generating the resulting labelled~$\DAGof{G}$ since we can reorder the numerators and denominators from~\wref{eq:prob_dist}.

\begin{lemma}[{{\cite[Lem.\,7]{LuczakMagnerSzpankowski2019}}}]
\label{lem:same-prob}
	Let $H \sim \PrefAttch$ for some $M \geq 1$. For any two graphs $G$, $G'$ without parallel edges apart from the seed graph $G_1$ satisfying $\UDAGof{G} = \UDAGof{G'}$, we have
	$\mathbb{P}(H = G) = \mathbb{P}(H = G')$.
\end{lemma}

For an unlabelled graph $S(G)$ without parallel edges, \wref{lem:same-prob} implies
	
\[
		\Prob{S(G)} 
	\wwrel= 
		\Prob{G} \cdot |\text{Adm}(S(G))|.
\]
Since $|\mathrm{Adm}(G)| \leq n!$, the instance-specific lower bound for unlabelled~$S(G)$ is
\[
		\lg\left(\frac{1}{\Prob{S(G)}}\right) 
	\wwrel= 
		\lg\left(\frac{1}{\Prob{G}}\right) - \lg|\text{Adm}(S(G))| 
	\wwrel\ge 
		\lg(1/\Prob{G}) - \lg(n!),
\]
which proves \wref{thm:lower-bound-ulbl}.

    \section{Space and Time Analysis of our Data Structure}
\label{app:space-analysis-ds}

In this section, we prove the space and time complexity of our preferential-attachment graph representation.
Note that for labelled graphs, the result follows directly from the guarantees of wavelet trees:
We use \wref{lem:wavelet-tree-simple} to represent $A=A(G)$; since we have (by definition) that $H_0(A(G)) = \DegreeEntropy(G)$, \wref{thm:data-structure-labeled} follows.

For \wref{thm:data-structure-unlabeled}, we first state the achieved results in terms of the empirical entropy of $A'$:

\begin{lemma}
\label{lem:ds-S-T}
    There exist data structures to represent $T$ and $A'$ in $H_0(A')+2n+o(Mn)$ bits of space, while allowing for $\mathsf{N_{out}}(v,i)$, $\mathsf{N_{in}}(v,i)$, $\mathsf{degree}(v)$, and $\mathsf{adjacency}(u,v)$ to run in $\mathrm{O}(\lg n)$ time, and $\mathsf{N_{out}}(v)$ and $\mathsf{N_{in}}(v)$, in $\mathrm{O}(\lg n)$ time per neighbour.
\end{lemma}

\begin{proof}
    As stated in \wref{lem:tree-succinct}, the tree $T$ can be succinctly represented using $2n + o(n)$ bits, while supporting relevant tree operations in constant time. Since $T$ encodes the first out-neighbour of each vertex in $G$, the remaining $n(M-1)$ edges in $A'$ must be stored separately.

    We store $A'$ using a wavelet tree. By \wref{lem:wavelet-tree-simple}, this uses $H_0(A') + o(Mn)$ bits (since $|A'| = n(M-1)$), while allowing access, rank, and select queries in $\mathrm{O}(\lg \sigma)$ time, where $\sigma$, the alphabet size, is $n$ here.

    Overall, the data structures representing the preferential-attachment graph $G$ occupy a total of $H_0(A') + 2n + o(Mn)$ bits.

    As established in \wref{lem:wavelet-tree-simple} and \wref{lem:tree-succinct}, the operations $\mathsf{parent}(T,v)$, $\mathsf{child}(T,v,i)$, and $\mathsf{degree}(T,v)$ are supported in $\mathrm{O}(1)$ time, while $\mathsf{access}(S,i)$, $\mathsf{rank}_v(S,i)$, and $\mathsf{select}_v(S,i)$ run in $\mathrm{O}(\lg n)$ time. Consequently, $\mathsf{N_{out}}(v,i)$, $\mathsf{N_{in}}(v,i)$, $\mathsf{degree}(v)$, and $\mathsf{adjacency}(u,v)$ can be computed in $\mathrm{O}(\lg n)$ time, while $\mathsf{N_{out}}(v)$ and $\mathsf{N_{in}}(v)$ require $\mathrm{O}(\lg n)$ time per neighbour.
\end{proof}

Towards the proof of \wref{thm:data-structure}, we now need to connect $H_0(A')$ and $\DegreeEntropy(S) = H_0(A)$.
We start observing that the trivial bound for the empirical entropy, $H_0(A') \le (M-1)n \lg n$, yields the second term in \wref{thm:data-structure-unlabeled}. 
Since this is also the expected value of our lower bound (computed by \L uczak et al.~\cite{LuczakMagnerSzpankowski2019}), for typical graphs, our data structure uses asymptotically the optimal space for an unlabeled preferential-attachment graph.

For a closer analysis, it is convenient to consider the 
\emph{per-character empirical entropy}: for a string $w \in \Sigma^n$ of length $n=|w|$,
we define $H_0^{\mathit{pc}}(w) = \frac1n H_0(w) \le \lg |\Sigma|$.
We point out that in general, deleting a character from a string may increase or decrease $H_0^{\mathit{pc}}$.
In \wref{app:proof-conjecture}, we will show that our scheme of choosing $T$, however, never increases the per-character entropy:

\begin{lemma}
\label{lem:conjecture}
	Let $G$ be any $M$-out-regular DAG $G$ and let $A=A(G)$ and $A'$ be as per our construction. 
	Then $H_0^{\mathit{pc}}(A') \leq H_0^{\mathit{pc}}(A)$.
\end{lemma}
Using \wref{lem:conjecture}, the space from \wref{lem:ds-S-T} becomes
\begin{align*}
		H_0(A') + 2n + o(Mn)
	&\wwrel=
		H_0^{\mathit{pc}}(A')|A'| + 2n + o(Mn)
\\	&\wwrel=
		H_0^{\mathit{pc}}(A')(M-1)n + 2n + o(Mn)
\\	&\wwrel\le
		H_0^{\mathit{pc}}(A(S))(M-1)n + 2n + o(Mn)
\\	&\wwrel=
		H_0(A(S))\frac{(M-1)n}{Mn} + 2n + o(Mn)
\\	&\wwrel=
		\DegreeEntropy(S)\bigl(1-\tfrac1M\bigr) + 2n + o(Mn).
\end{align*}
This concludes the proof of \wref{thm:data-structure}.

    \section[Proof of Lemma \protect\ref*{lem:conjecture} (Tree Deletion)]{Proof of \wref{lem:conjecture} (Tree Deletion)}
\label{app:proof-conjecture}

\newcommand{\concat}{\cdot}
\newcommand\pcH{H_0^{\mathit{pc}}}

In this appendix, we prove \wref{lem:conjecture}, the key lemma for bounding the space of the unlabelled graph data structure.
We will first reformulate the problem as a question on strings, which we consider natural enough to be of independent interest.

\subsection{Blocked Character Deletion}

Let $X \concat Y$ denote the concatenation of two strings $X$ and $Y$, and let $X \setminus c$, for a character $c$, denote the result of deleting a single occurrence of $c$ from $X$. 

With that, we can translate the process by which $A'$ results from $A$ to strings:
Given a string $A=X_1 \concat X_2 \concat \dots \concat X_n$ of $n$ \emph{blocks} (\ie\ substrings) and alphabet $\Sigma$, we obtain string $A'=X'_1 \concat X'_2 \concat \dots \concat X'_n$ where, for $i \in [n]$, $X'_i = X_i \setminus c_i$, for $c_i \in X_i$; \ie,
$A'$ is obtained by deleting exactly one character in each block of $A$.
Can we achieve $\pcH(A') \le \pcH(A)$?
Our goal is to show that a simple greedy method, the LFC scheme below, suffices for that, 
provided all blocks are equally long, $|X_1| = |X_2| = \cdots = |X_n| = M$.

\subsection{Least-Frequent-Character (LFC) Scheme}
\label{s:scheme}

We will now describe an algorithm that corresponds to our construction of the tree $T$ from 
\wref{sec:data-structures}.

First, let $S$ be a sorted copy of $A$, where we sort characters by increasing frequency.
Define the bijective function $\sigma : \Sigma \rightarrow [|\Sigma|]$ which maps each character of the alphabet $\Sigma$ of $A$ to a \emph{rank} such that, for distinct characters $x,y\in \Sigma$, $|A|_x > |A|_y \implies \sigma(x) > \sigma(y)$ (such a mapping must exist). Then, construct the string $S$ as follows: start with $S=\emptyset$, and for each $i$ from $1$ to $|\Sigma|$, append the character $\sigma^{-1}(i)$ $|A|_{\sigma^{-1}}$ times.

\begin{example}
    \label{ex:defaap}
    Let $n=3$, $M=4$, and $A=\texttt{abracadabraa}$. We can then define 
    \[\sigma = \{(\texttt{c},1),(\texttt{d},2),(\texttt{b},3), \\
    (\texttt{r},4),
    (\texttt{a},5)\},\]
and construct $S=\texttt{cdbbrraaaaaa}$. Note that there can be multiple ways to define $\sigma$.
\end{example}

\noindent To construct $A'$, we delete from each block its leftmost letter in $S$.
Equivalently, but more convenient for the analysis to follow,
assume without loss of generality that $\texttt{\lambda} \not\in \Sigma$ and follow the procedure described in \wref{alg:lfc}. 

\begin{algorithm}[tbhp]
\caption{Computation of $A'$ following the LFC scheme.}\label{alg:lfc}
\begin{algorithmic}[1]
\State $\hat{A} \gets A$ 
\Comment{$\hat{A}=\hat{X}_1 \concat \hat{X}_2 \concat \dots \concat \hat{X}_n$ is a copy of $A=X_1 \concat X_2 \concat \dots \concat X_n$.}
\State $F[1..n] \gets [0]^n$
\Comment{$F$ stores which blocks in $\hat{A}$ were flagged}
\For{$i=1$ to $n$}
    \State $k \gets \min \{k':S[k']\not= \texttt{\lambda}\}$
    \Comment{$k$ is set to the smallest index such that $S[k] \not= \texttt{\lambda}$.}
    \State $c \gets S[k]$
    \Comment{$c$ is set to the character at that index.}
    \State $j \gets \min \{j':c \in \hat{X}_{j'} \textnormal{ and } F[j'] = 0\}$
    \Comment{Index of some unflagged block where $c$ occurs.}
    \State $F[j] \gets 1$
    \Comment{We register that the block has been flagged.}
    \ForAll{$l \in \hat{X}_j$}
    \Comment{We iterate over all characters in the selected block $\hat{X}_j$.}
        \State $i' \gets \min \{i'':S[i'']= l \}$
        \State $S[i'] \gets \texttt{\lambda}$ 
    \EndFor
    \State $k \gets \min \{k' : \hat{X}_j[k']=c\}$
    \State $\hat{X}_j[k] \gets \texttt{\lambda}$
    \Comment{This also updates $\hat{A}$ since $\hat{X}_j \subseteq \hat{A}$.}
\EndFor
\State $A' \gets \emptyset$
\For{$i \gets 1$ to $nM$}
    \If{$\hat{A}[i] \not= \texttt{\lambda}$}
        \State $A' \gets A' \concat \hat{A}[i]$
        \Comment{$A'$ is built by copying $\hat{A}$ and deleting positions where $\texttt{\lambda}$ occurs.}
    \EndIf
\EndFor
\State \Return $A'$
\end{algorithmic}
\end{algorithm}

Initially, all blocks in $A$ are unflagged. The scheme proceeds in $n$ steps. In each step, it selects the leftmost non-$\texttt{\lambda}$ character in $S$ (call it $c$) located at position $k$ (\ie\ $S[k] = c \ne \texttt{\lambda}$), and flags every unflagged block that contains $c$; suppose it flags $p$ such blocks. Then, in each of the $p$ flagged blocks, it replaces exactly one occurrence of $c$ in the block with $\texttt{\lambda}$. For each remaining character $c'$ in the block, it replaces one occurrence of $c'$ in $S$ with $\texttt{\lambda}$. Finally, it replaces $c$ itself with $\texttt{\lambda}$.

At each step, exactly $M$ characters are replaced by $\texttt{\lambda}$ in $S$, so that after $n$ steps, all $nM$ characters in $S$ will have been replaced by $\texttt{\lambda}$. The string $A'$ is obtained by deleting all occurrences of $\texttt{\lambda}$ in $A$ following the execution of the scheme.

\begin{example}
    Let $n=3$, $M=4$, and $A=\texttt{abracadabraa}$. Define $\sigma$ as in Example~\ref{ex:defaap}; we then get $S=\texttt{cdbbrraaaaaa}$. Let us compute $A'$ following the scheme outlined in Algorithm~\ref{alg:lfc}.

\newcommand*{\mybox}[2][5cm]{%
  \makebox[#1][s]{#2}}

\begin{figure}[tbhp]
    \centering
    \bigskip
\begin{center}
    \begin{tabular}{cccc} \toprule
    \# \textbf{for} & $\hat{A}$ & $S$ & $F$ \\ \midrule
     & $\texttt{abracadabraa}$ & $\texttt{cdbbrraaaaaa}$ & $[0,0,0]$ \\
    $1$ & $\texttt{abra\textcolor{red}{\lambda ada}braa}$ & $\texttt{\lambda\lambda bbrr\lambda\lambda aaaa}$ & $[0,1,0]$ \\
    $2$ & $\textcolor{red}{\texttt{a\lambda ra}}\texttt{\lambda adabraa}$ & $\texttt{\lambda\lambda\lambda br\lambda\lambda\lambda\lambda\lambda\texttt{aa}}$ & $[1,1,0]$ \\
    $3$ & $\texttt{a\lambda ra\lambda ada\textcolor{red}{\lambda\texttt{raa}}}$ & $\texttt{\lambda\lambda\lambda\lambda\lambda\lambda\lambda\lambda\lambda\lambda\lambda\lambda}$ & $[1,1,1]$ \\ \bottomrule
    \end{tabular}
\end{center}
    \caption{Sample execution of the LFC scheme on the string $A=\texttt{abracadabraa}$. The resulting $A'$ is \texttt{araadaraa}, and $\pcH(A') \approx 1.22439 \leq \pcH(A) \approx 1.95915$.}
    \label{fig:enter-label}
\end{figure}
\end{example}

\subsection{Analysis}
\label{sec:correctness}

The proof uses a handy tool to bound entropies: the \emph{majorisation} partial order.
Let $P=(p_1,p_2, \dots, p_{n})$ and $Q=(q_1,q_2, \dots, q_{n})$ be two $n$-sized distributions. We say that $P$ \emph{majorises} $Q$ (i.e. $Q \preceq P$) if and only if $\sum_{i=1}^{k} p_i^{\downarrow} \geq \sum_{i=1}^{k} q_i^{\downarrow}$ for all $k \in [n]$, where $X^{\downarrow}=(x_1^{\downarrow}, x_2^{\downarrow}, \dots, x_{|X|}^{\downarrow})$ denotes the $|X|$-sized vector of the distribution $X$ sorted in non-increasing order.
Given that the entropy function is \emph{Schur-concave}, $Q \preceq P$ implies $\pcH(P) \leq \pcH(Q)$~\cite{marshall1979inequalities}.

Assume $M \geq 2$; the case where $M=1$ is trivial. Let $S$ be a string over alphabet $\Sigma$, and let $P=(p_1,p_2, \dots, p_{|\Sigma|})$ be a (non-ordered) distribution such that $p_i= \frac{|S|_{c_i}}{|S|}$ for all $i \in [|\Sigma|]$, where $c_i \in \Sigma$ is the $i$th character of the alphabet (without loss of generality, assume the latter is ordered).

\begin{lemma}
\label{lm:sorting}
    Let $P=(p_1,p_2, \dots, p_{n})$ and $Q=(q_1,q_2, \dots, q_{n})$ be two $n$-sized distributions such that $\sum_{i=1}^{k} p_i \geq \sum_{i=1}^{k} q_i^{\downarrow}$ for all $k \in [n]$. Then, $\sum_{i=1}^{k} p_i^{\downarrow} \geq \sum_{i=1}^{k} q_i^{\downarrow}$ for all $k \in [n]$.
\end{lemma}

\begin{proof}
    Sorting $P$ in non-increasing order does not reduce prefix sums, thus giving $\sum_{i=1}^k p_i^{\downarrow} \geq \sum_{i=1}^k p_i \geq \sum_{i=1}^k q^{\downarrow}_i$ for all $k \in [n]$.
\end{proof}

\begin{lemma}
    Let $k_1 < k_2 < \dots < k_n$ be the $n$ indices in $S$ picked by the LFC scheme (Algorithm~\ref{alg:lfc}, line 4). We have $k_i \leq (i-1)M+1$.
\end{lemma}

\begin{proof}
    Suppose, towards a contradiction, that at the $i$th selection, the leftmost non-$\texttt{\lambda}$ character in $S$ is at position $(i-1)M+2$ or greater. This implies that for all $j \in [(i-1)M+1]$, $S[j]=\texttt{\lambda}$.

    Since $i-1$ characters were selected previously, and $(M-1)(i-1)$ other characters were replaced by $\texttt{\lambda}$ each time, then the next non-$\texttt{\lambda}$ character must be at position at most $(M-1)(i-1)+(i-1)+1=M(i-1)+1$. However, we assumed that position was occupied by $\texttt{\lambda}$, which is a contradiction.
\end{proof}

\begin{lemma}
    \label{lemma:defaultisgood}
    Let $A$ be a string of size $nM$ and alphabet $\Sigma$, and construct its respective $S$ following the LFC scheme. Let $S'$ be obtained by copying $S$ and simultaneously deleting its characters at positions $(i-1)M+1$ for all $i \in [n]$. We then have $\pcH(S') \leq \pcH(S)$.
\end{lemma}
\begin{proof}
    Let $f_1 \leq f_2 \leq \dots f_{|\Sigma|}$ be the normalised frequencies of the characters of $\Sigma$ in their order of appearance in $S$ (by the way $S$ is constructed, these must be non-decreasing), and let $f'_i$ be the respective normalised frequency of the $i$th distinct character of $S$ in $S'$; for instance, suppose $f_5$ is the frequency of the fifth distinct character $c$ in $S$, then $f'_5$ is the frequency of $c$ in $S'$ (i.e. after the $n$ deletions in $S$). 
    
    The post-deletion frequency $f'_i$ can be expressed in terms of $f_i$ under two distinct cases:
    \begin{equation}
    \label{eq:freqexp}
    f'_k =
    \begin{dcases}
    \frac{nMf_i-\lceil nf_k\rceil}{n(M-1)},& \text{(case 1) if \,\,\,} \Biggl(n\sum_{i=1}^{k-1}f_i\Biggr)\bmod{M} \rel\geq M - (nf_k \bmod{M}) \\ \\
    \frac{nMf_i-\lfloor nf_k\rfloor}{n(M-1)},              & \text{(case 2) otherwise.}
    \end{dcases}
    \end{equation}

    \noindent From those expressions, we want to show the following equality:
    \begin{equation}
    \label{eq:freqsums}
                \sum_{i=1}^k f'_i = \sum_{i=1}^k f_i+\frac{(n\sum_{i=1}^k f_i)\bmod{M}}{nM(M-1)};
    \end{equation}
    to do so, we proceed by induction.
    
    \paragraph{Base case ($k=1$)} 
    Clearly, the second case of Equation~\ref{eq:freqexp} applies for (one of) the largest character(s) (of frequency $f_k$ in $S$), since there are no previous frequencies to add up, so $(n\sum_{i=1}^{0}f_i)\bmod{M} = 0 < M - (nf_1 \bmod{M})$ (the RHS can never be $0$). Hence, $f'_1 = \frac{Mnf_1-\lfloor nf_1 \rfloor}{n(M-1)} \geq f_1$. The last inequality holds because:
\begin{align*}
    \lfloor nf_1 \rfloor &\leq nf_1 \\
-M\lfloor nf_1 \rfloor &\geq -Mnf_1 \\
M^2nf_1-M\lfloor nf_1 \rfloor &\geq M^2nf_1-Mnf_1 \\
M(Mnf_1-\lfloor nf_1 \rfloor) &\geq Mnf_1(M-1) \\
Mnf_1-\lfloor nf_1 \rfloor &\geq nf_1(M-1) \\
\frac{Mnf_1-\lfloor nf_1 \rfloor}{n(M-1)} &\geq \frac{nf_1(M-1)}{n(M-1)} \\
\frac{Mnf_1-\lfloor nf_1 \rfloor}{n(M-1)} &\geq f_1.
\end{align*}

\paragraph{Induction hypothesis} 

Assume that 
\begin{equation*}
    \sum_{i=1}^k f'_i = \sum_{i=1}^k f_i+\frac{(n\sum_{i=1}^k f_i)\bmod{M}}{(M-1)Mn}
\end{equation*}
for some $k$. Let us show that the equality holds when adding $f'_{k+1}$ in either of the two cases of Equation~\ref{eq:freqexp}.

\paragraph{Case 1} 
We delete the character $\lceil nf_{k+1} \rceil$ times, i.e.
\begin{equation*}
    f'_{k+1} = \frac{Mnf_{k+1}-\lceil nf_{k+1} \rceil}{n(M-1)}.
\end{equation*}

\noindent Recall that this case occurs if and only if
\begin{equation*}
    \bigg(n\sum_{i=1}^{k}f_i\bigg)\bmod{M} \wrel\geq M - (nf_{k+1} \bmod{M}).
\end{equation*}

\noindent Therefore, we have
\begin{equation*}
    \frac{((n\sum_{i=1}^{k}f_i)\bmod{M})}{n(M-1)} \geq \frac{M - (f_{k+1}n \bmod{M})}{n(M-1)}
\end{equation*}
and since
\begin{equation*}
    M-(nf_{k+1} \bmod{M}) = M(\lceil nf_{k+1}\rceil-nf_{k+1}),
\end{equation*}
then
\begin{equation*}
    \frac{((n\sum_{i=1}^{k}f_i)\bmod{M})}{n(M-1)} \geq \frac{ M(\lceil nf_{k+1}\rceil-nf_{k+1})}{n(M-1)}.
\end{equation*}

\noindent Now, back to the assumed equality. Let us add the new frequency on both sides. We thus get
\begin{equation*}
    \sum_{i=1}^k f'_i + f'_{k+1}= \sum_{i=1}^k f_i+\frac{(n\sum_{i=1}^k f_i)\bmod{M}}{(M-1)Mn}+ f'_{k+1},
\end{equation*}
hence giving us
\begin{equation*}
    \sum_{i=1}^{k+1} f'_i= \sum_{i=1}^k f_i+\frac{(n\sum_{i=1}^k f_i)\bmod{M}}{(M-1)Mn}+ \frac{Mnf_{k+1}-\lceil nf_{k+1} \rceil}{n(M-1)}.
\end{equation*}

\noindent By splitting the fraction, we get:
\begin{equation*}
    \sum_{i=1}^{k+1} f'_i= \sum_{i=1}^k f_i+\frac{(n\sum_{i=1}^k f_i)\bmod{M}}{(M-1)Mn}+ \frac{nf_{k+1}-\lceil nf_{k+1} \rceil}{n(M-1)}+\frac{n(M-1)f_{k+1}}{n(M-1)}
\end{equation*}
which gives us, by cancelling out the factors in the last fraction,
\begin{align*}
    \sum_{i=1}^{k+1} f'_i &= \sum_{i=1}^k f_i+\frac{(n\sum_{i=1}^k f_i)\bmod{M}}{(M-1)Mn}+ \frac{nf_{k+1}-\lceil nf_{k+1} \rceil}{n(M-1)}+f_{k+1} \\
    \sum_{i=1}^{k+1} f'_i &= \sum_{i=1}^{k+1} f_i+\frac{(n\sum_{i=1}^{k} f_i)\bmod{M}}{(M-1)Mn}+ \frac{nf_{k+1}-\lceil nf_{k+1} \rceil}{n(M-1)}.
\end{align*}

\noindent By multiplying the numerator and denominator by $M$ in the last summand, we obtain
\begin{align*}
    \sum_{i=1}^{k+1} f'_i= \sum_{i=1}^{k+1} f_i+\frac{(n\sum_{i=1}^{k} f_i)\bmod{M}}{(M-1)Mn}+ \frac{M(nf_{k+1}-\lceil nf_{k+1} \rceil)}{Mn(M-1)}
    \\
    \sum_{i=1}^{k+1} f'_i= \sum_{i=1}^{k+1} f_i+\frac{(n\sum_{i=1}^{k} f_i)\bmod{M}}{(M-1)Mn}- \frac{M(\lceil nf_{k+1} \rceil -nf_{k+1})}{Mn(M-1)}.
\end{align*}

\noindent Note that we can already establish that $\sum_{i=1}^{k+1} f'_i \geq \sum_{i=1}^{k+1} f_i$ since
\begin{equation*}
    \frac{(n\sum_{i=1}^{k} f_i)\bmod{M}}{(M-1)Mn} \geq \frac{M(\lceil nf_{k+1} \rceil -nf_{k+1})}{Mn(M-1)} \geq 0.
\end{equation*}

\noindent Finally, we have
\begin{align*}
    \sum_{i=1}^{k+1} f'_i &= \sum_{i=1}^{k+1} f_i+\frac{(n\sum_{i=1}^{k} f_i)\bmod{M}}{(M-1)Mn}- \frac{M-(nf_{k+1} \bmod{M})}{Mn(M-1)}
\\
    \sum_{i=1}^{k+1} f'_i &= \sum_{i=1}^{k+1} f_i+\frac{(n\sum_{i=1}^{k+1} f_i)\bmod{M}}{(M-1)Mn},
\end{align*}
where the last equality holds because
\begin{equation*}
    \bigg(n\sum_{i=1}^{k+1} f_i\bigg)\bmod{M}+(nf_{k+1})\bmod{M} \geq M,
\end{equation*}
which strictly follows from the condition of the case.

\paragraph{Case 2} 
We delete the character $\lfloor nf_{k+1} \rfloor$ times, i.e.
\begin{equation*}
    f'_{k+1} = \frac{Mnf_{k+1}-\lfloor nf_{k+1} \rfloor}{n(M-1)}.
\end{equation*}

\noindent Again, let us add the new frequency on both sides. We have
\begin{align*}
    \sum_{i=1}^k f'_i + f'_{k+1} &= \sum_{i=1}^k f_i+\frac{(n\sum_{i=1}^k f_i)\bmod{M}}{(M-1)Mn}+ f'_{k+1} \\
    \sum_{i=1}^{k+1} f'_i &= \sum_{i=1}^k f_i+\frac{(n\sum_{i=1}^k f_i)\bmod{M}}{(M-1)Mn}+ \frac{Mnf_{k+1}-\lfloor nf_{k+1} \rfloor}{n(M-1)}.
\end{align*}

\noindent By splitting the fraction, we get
\begin{align*}
    \sum_{i=1}^{k+1} f'_i &= \sum_{i=1}^k f_i+\frac{(n\sum_{i=1}^k f_i)\bmod{M}}{(M-1)Mn}+ \frac{nf_{k+1}-\lfloor nf_{k+1} \rfloor}{n(M-1)}+\frac{n(M-1)f_{k+1}}{n(M-1)} \\
    \sum_{i=1}^{k+1} f'_i &= \sum_{i=1}^k f_i+\frac{(n\sum_{i=1}^k f_i)\bmod{M}}{(M-1)Mn}+ \frac{nf_{k+1}-\lfloor nf_{k+1} \rfloor}{n(M-1)}+f_{k+1} \\
    \sum_{i=1}^{k+1} f'_i &= \sum_{i=1}^{k+1} f_i+\frac{(n\sum_{i=1}^k f_i)\bmod{M}}{(M-1)Mn}+ \frac{nf_{k+1}-\lfloor nf_{k+1} \rfloor}{n(M-1)}.
\end{align*}

\noindent By multiplying the numerator and denominator by $M$ in the last summand, we obtain
\begin{align*}
    \sum_{i=1}^{k+1} f'_i &= \sum_{i=1}^{k+1} f_i+\frac{(n\sum_{i=1}^{k} f_i)\bmod{M}}{(M-1)Mn}+ \frac{M(nf_{k+1}-\lfloor nf_{k+1} \rfloor)}{Mn(M-1)} \\
    \sum_{i=1}^{k+1} f'_i &= \sum_{i=1}^{k+1} f_i+\frac{(n\sum_{i=1}^{k} f_i)\bmod{M}+M(nf_{k+1}-\lfloor nf_{k+1} \rfloor)}{(M-1)Mn} \\
    \sum_{i=1}^{k+1} f'_i &= \sum_{i=1}^{k+1} f_i+\frac{(n\sum_{i=1}^{k} f_i)\bmod{M}+nf_{k+1} \bmod{M}}{(M-1)Mn} \\
    \sum_{i=1}^{k+1} f'_i &= \sum_{i=1}^{k+1} f_i+\frac{(n\sum_{i=1}^{k+1} f_i)\bmod{M}}{(M-1)Mn},
\end{align*}
where the last equation strictly follows from the condition of the case, i.e.
\begin{equation*}
    \bigg(n\sum_{i=1}^{k+1} f_i\bigg)\bmod{M}+(nf_{k+1})\bmod{M} < M.
\end{equation*}

\noindent In either case, the equality holds. Notice that we have, for all $k \in [|\Sigma|]$:
\[
\sum_{i=1}^{k} f'_i \geq \sum_{i=1}^{k} f_i
\]
since $\frac{(n\sum_{i=1}^k f_i)\bmod{M}}{(M-1)Mn}$ is non-negative.

Let $P=(f_1,f_2, \dots, f_{|\Sigma|})$ and $P'=(f'_1,f'_2, \dots, f'_{|\Sigma|})$. Clearly, by Lemma~\ref{lm:sorting} (note that $P=P^{\downarrow}$), the last inequality implies that $P^{\downarrow} \preceq P'^{\downarrow}$. This further implies that $\pcH(P'^{\downarrow}) \leq \pcH(P^{\downarrow})$, and the lemma follows.
\end{proof}

Let $k_1 < k_2 < \dots < k_n$ be the indices picked by the LFC scheme. Given that $k_i \leq (i-1)M+1$ for $i \in [n]$, then any set of indices from the scheme can be obtained by setting all the indices to their respective upper bounds and \emph{shifting} them to the left (i.e. decreasing their respective $k$).

\begin{lemma}
\label{lemma:shiftingisokay}
    Let $k_1 < k_2 < \dots < k_n$ be the indices selected by the LFC scheme and let
    \[
\sum_{i=1}^{j} f'_i \geq \sum_{i=1}^{j} f_i
\] 
for all $j \in [n]$. If some $k_i$ is set to $k_i-x$, with $x \in \mathbb{N}$, such that $k_i \geq 1$ and $k_i \not= k_j$ for all $j \in [n]$, then the inequality still holds.
\end{lemma}
\begin{proof}
    Suppose $S[k_i]=S[k_i-x]$. Clearly, \emph{shifting} the index to the left does not change the frequency of any character in $\Sigma$; therefore, in that case, the inequality will still hold.

    Now, suppose $S[k_i]\not=S[k_i-x]$ and let $f'_a$ and $f'_b$ be the impacted frequencies, with $a < b$. Shifting the index is equivalent to restoring the character at position $k_i$ and deleting the character at position $k_i-x$. Thus, $f'_b$ decreases by some value $\delta$, and $f'_a$ increases by $\delta$. This clearly does not affect the prefix sums, and the lemma follows.
\end{proof}

\begin{proofof}{\wref{lem:conjecture}}
    Construct $S$, initially set the ``deleting'' indices to $1, M+1, \dots, (n-1)M+1$, and let $S'$ be the string obtained by deleting the characters of $S$ at those indices. By Lemma~\ref{lemma:defaultisgood}, $\pcH(S') \leq \pcH(S)$. We know that any set of indices computed by the scheme on $A$ must result with indices $k_i \leq (i-1)M+1$; so, for each index $k_i$, shift them to wherever the LFC scheme would put them on input $A$. That shifting still does not increase the entropy beyond that of $S$, by Lemma~\ref{lemma:shiftingisokay}; thus $\pcH(S') \leq \pcH(S)$ holds after the shifts.
    
    Since $\pcH(A)=\pcH(S)$, $\pcH(S') = \pcH(A')$ and $\pcH(S') \leq \pcH(S)$, it follows that $\pcH(A') \leq \pcH(A)$, completing the proof.
\end{proofof}

\begin{remark}[Equal-sized blocks needed]
	We point out that, in contrast to \wref{lem:conjecture}, 
	allowing blocks of \emph{different} lengths in blocked character deletion
	can make it impossible to achieve $\pcH(A')\le \pcH(A)$.
\end{remark}

%}{}

%
%

%%
%
%
%
%
%
%
%
%
%
%
%
%
%
%
%
%
%
%
%
%
%
%
%
%
%
%
%
%
%
%
%
%
%
%
%
%
%
%
%
%
%
%
%
%
%
%
%
%
%
%%
%%
%%
%
%
%
%
%%
%%
%%
%%
%%
%%
%%
%%
%
%
%
%
%
%
%
%
%
%
%
%

%
%
%
%
%

\end{document}